\documentclass[11pt]{article}%
\usepackage{fullpage}
\usepackage{amssymb,amsmath}
\usepackage{graphicx, epsfig}
\usepackage{caption}
\usepackage{subcaption}

\def\idrm#1{\ensuremath{\mathrm{#1}}}
\def\idtt#1{\ensuremath{\mathtt{#1}}}

\newtheorem{theorem}{Theorem}
\newtheorem{lemma}{Lemma}
 \newtheorem{fact}{Fact}

 \newenvironment{proof}{\trivlist\item[]\emph{Proof}:}%
 {\unskip\nobreak\hskip 1em plus 1fil\nobreak$\Box$
 \parfillskip=0pt%
 \endtrivlist}

\newenvironment{itemize*}%
  {\begin{itemize}%
    \setlength{\itemsep}{0pt}%
    \setlength{\parskip}{0pt}%
    \setlength{\parsep}{0pt}%
    \setlength{\topsep}{0pt}%
    \setlength{\partopsep}{0pt}%
  }%
  {\end{itemize}}%

\newenvironment{enumerate*}%
  {\begin{enumerate}%
    \setlength{\itemsep}{0pt}%
    \setlength{\parskip}{0pt}%
    \setlength{\parsep}{0pt}%
    \setlength{\topsep}{0pt}%
    \setlength{\partopsep}{0pt}%
  }%
  {\end{enumerate}}%

\DeclareMathOperator\Max{Max}
\DeclareMathOperator\Min{Min}
\DeclareMathOperator\Nav{Nav}

\newcommand{\no}[1]{}
\newcommand{\myparagraph}[1]{\textbf{#1}}

\newcommand{\cP}{{\cal P}}
\newcommand{\cS}{{\cal S}}
\newcommand{\cT}{{\cal T}}
\newcommand{\cC}{{\cal C}}
\newcommand{\cD}{{\cal D}}

\newcommand{\cM}{{\cal M}}
\newcommand{\cB}{{\cal B}}

\newcommand{\cW}{{\cal W}}

\newcommand{\parent}{\mathit{par}}
\newcommand{\eps}{\varepsilon}
\newcommand{\orth}{_{\mbox{\scriptsize\rm orth}}}
\newcommand{\spred}{\idrm{pred}}
\newcommand{\ssucc}{\idrm{succ}}
\newcommand{\lab}{\idtt{lab}}

\newcommand{\shortver}[1]{}
\newcommand{\longver}[1]{#1}
\newcommand{\shlongver}[2]{#2}

\begin{document}

\title{Dynamic Planar Point Location in External Memory}
\author{
  J. Ian Munro\thanks{Cheriton School of Computer Science, University of Waterloo. Email {\tt imunro@uwaterloo.ca}.}
  \and 
  Yakov Nekrich\thanks{Cheriton School of Computer Science, University of Waterloo.
    Email: {\tt yakov.nekrich@googlemail.com}.}
}
\date{}

\maketitle
\begin{abstract}
  In this paper we describe a fully-dynamic data structure for the planar point location problem
 in the external memory model. Our data structure supports queries in  $O(\log_B n(\log\log_B n)^3))$ I/Os and updates in $O(\log_B n(\log\log_B n)^2))$ amortized I/Os, where $n$ is the number of segments in the subdivision and $B$ is the block size.  This is the first dynamic data structure with almost-optimal query cost.  For comparison all previously known results for this problem require  $O(\log_B^2 n)$ I/Os to answer queries. Our result almost matches  the best known upper bound in the internal-memory model.
\end{abstract}

\thispagestyle{empty}
\newpage
\setcounter{page}{1}
\section{Introduction}
\label{sec:intro}
Planar point location is a classical computational geometry problem with a number of important applications. In this problem we keep a polygonal subdivision $\Pi$ of the two-dimensional plane in a data structure; for an arbitrary query point $q$, we must be able to find the face of $\Pi$ that contains $q$.  In this paper we study the dynamic version of this problem in the external memory model.  We show that a planar subdivision can be maintained under insertions and deletions of edges, so that the cost of queries and updates is close to $O(\log_B n)$, where $n$ is the number of segments in the subdivision and $B$ is the block size.

Planar point location problem was studied extensively in different computational models. 
Dynamic internal-memory data structures for general subdivisions were described by Bentley~\cite{Bentley}, Cheng and Janardan~\cite{ChengJ92}, Baumgarten et al.~\cite{BaumgartenJM94}, Arge et al.~\cite{ArgeBG06}, and Chan and Nekrich~\cite{ChanN15}.  Table~\ref{table:intmem} lists previous results. We did not include in this table many other results for special cases of the point location problem, such as the data structures for monotone, convex, and orthogonal subdivisions, e.g.,~\cite{PreparataT89,PreparataT92,ChiangT92,ChiangPT96,GoodrichT98,GioraK,ChanT18}.  The currently best data structure~\cite{ChanN15} achieves\footnote{In this paper $\log n$ denotes the binary logarithm of $n$ when the logarithm base is not specified.} $O(\log n)$ query time and $O(\log^{1+\eps} n)$ update time or $O(\log^{1+\eps} n)$ query time and $O(\log n)$ update time; the best query-update trade-off described in~\cite{ChanN15} is $O(\log n\log\log n)$ randomized query time and $O(\log n\log\log n)$ update time. See Table~\ref{table:intmem}.

In the external memory model~\cite{AggarwalV88} the data can be stored in the internal memory of size $M$ or on the external disk. Arithmetic operations can be performed only on data in the internal memory. Every input/output operation (I/O) either reads a block of $B$ contiguous words from the disk into the internal memory or writes $B$ words from the internal memory into disk. Measures of efficiency in this model are the number of I/Os needed to solve a problem and the amount of used disk space.

Goodrich et al.~\cite{GoodrichTVV93} presented a linear-space static external data structure for point location in a monotone subdivision with $O(\log_B n)$ query cost. Arge et al.~\cite{ArgeDT03} designed a  data structure for a general subdivison with the same query cost. Data structures for answering a batch of point location queries were considered in~\cite{GoodrichTVV93} and~\cite{ArgeVV07}. 
Only three external-memory results are known for the dynamic case. The data structure of Agarwal, Arge, Brodal, and Vitter~\cite{AgarwalABV99} supports queries on monotone subdivisions in $O(\log_B^2 n)$ I/Os and updates in $O(\log^2_B n)$ I/Os amortized. Arge and Vahrenhold~\cite{ArgeV04}  considered the case of general subdivisons; they retain the same cost for queries and insertions as~\cite{AgarwalABV99} and reduce the deletion cost to $O(\log_B n)$. Arge, Brodal, and Rao~\cite{ArgeBR12} reduced the insertion cost to $O(\log_B n)$. Thus all previous dynamic data structures did not break $O(\log^2_B n)$ query cost barrier. For comparison the first internal-memory data structure with query time close to logarithmic was presented by Baumgarten et al~\cite{BaumgartenJM94} in 1994.  See Table~\ref{table:extmem}. 
All previous data structures use $O(n)$ words of space (or $O(n/B)$ blocks of $B$ words\footnote{Space usage of external-memory data structures is frequently measured in disk blocks of $B$ words. In this paper we measure the space usage in words. But the  space usage of $O(n)$ words is equivalent to $O(n/B)$ blocks of space.}).

\begin{table}[tb]
  \centering
  \begin{tabular}{|c|c|l|l|l|l|} \hline
    Reference & Space & Query Time & Insertion Time & Deletion Time & \\ \hline
    Bentley \cite{Bentley} & $n\log n$ & $\log^2 n$ & $\log^2 n$
& $\log^2 n$ & \\
    Cheng--Janardan \cite{ChengJ92} & $n$ & $\log^{2} n$ & $\log n$ & $\log n$ & \\
    Baumgarten et al. \cite{BaumgartenJM94} & $n$ & $\log n\log\log n$ & $\log n\log\log n$  & $\log^2 n$  & $^\dagger$\\
    Arge et al. \cite{ArgeBG06} & $n$ & $\log n$ & $\log^{1+\eps}n$  & $\log^{2+\eps}n$   &  $^\dagger$\\
    Arge et al. \cite{ArgeBG06} & $n$ & $\log n$ & $\log n (\log \log n)^{1+\eps}$\! & $\log^2n/\log\log n$\! & $^{\dagger\ddagger\ast}$\!\\  \hline
    Chan and Nekrich \cite{ChanN15}  & $n$ & $\log n(\log\log n)^2$\! & $\log n\log\log n$ & $\log n\log\log n$ &\\
    Chan and Nekrich \cite{ChanN15}  & $n$ & $\log n$ & $\log^{1+\eps}n$ & $\log^{1+\eps}n$ &\\
    Chan and Nekrich \cite{ChanN15}  & $n$ & $\log n$ & $\log^{1+\eps}n$ & $\log n(\log\log n)^{1+\eps}$\! &$^{\ast}$\\
    Chan and Nekrich \cite{ChanN15}   & $n$ & $\log^{1+\eps} n$ & $\log n$ & $\log n$ &\\
    Chan and Nekrich \cite{ChanN15}   & $n$ & $\log n\log\log n$\! & $\log n\log\log n$ & $\log n\log\log n$ &$^{\ddagger\ast}\!$\\
 \hline
  \end{tabular}
  \caption{\rm Previous results on dynamic planar point location in internal memory. Entries marked $^\dagger$ and $^\ddagger$ require amortization and (Las Vegas) randomization respectively, $\eps>0$ is an arbitrarily small constant.  Results marked $^\ast$ are in the RAM model, all other results are in the pointer machine model. Space usage is measured in words.
}\label{table:intmem}
\end{table}

\begin{table}[t]
  \centering
\resizebox{\textwidth}{!}{%
  \begin{tabular}{|c|c|l|l|l|l|} \hline
    Reference & Space & Query Cost & Insertion Cost & Deletion Cost & \\ \hline
Agarwal et al \cite{AgarwalABV99} & $n$ & $\log_B^2 n$ & $\log_B^2 n$ & $\log_B^2 n$ & M\\
Arge and Vahrenhold \cite{ArgeV04}  & $n$ & $\log_B^2 n$ & $\log_B^2 n$ & $\log_B n$ & G\\
Arge et al \cite{ArgeBR12}  & $n$  & $\log_B^2 n$ & $\log_B n$ & $\log_B n$ & G\\ \hline
This paper  & $n$ & $\log_B n (\log\log_B n)^3$ & $\log_B n (\log\log_B n)^2$ & $\log_B n (\log\log_B n)^2$  & G\\
This paper  & $n$ & $\log_B n \log\log_B n$ & $\log_B n \log\log_B n$ & $\log_B n \log\log_B n$  & O\\ 
\hline
\end{tabular}
}
\caption{Previous and new results on dynamic planar point location in external memory. G denotes most general subdivisions, M denotes monotone subdivision, and O denotes orthogonal subdivision. Space usage is measured in words and update cost is amortized.}\label{table:extmem}
\end{table}

In this paper we show that it is possible to break the $O(\log^2_B n)$ barrier for the dynamic point location problem.
Our data structure answers queries in $O(\log_B n (\log\log_B n)^3)$ I/Os, supports updates in $O(\log_B n(\log\log_B n)^2)$ I/Os amortized, and uses linear space.  Thus we achieve close to logarithmic query cost and a query-update trade-off  almost matching the state-of-the-art upper bounds in the internal memory model. Our result is  within double-logarithmic factors from optimal. Additionally we describe a data structure that supports point location queries in an orthogonal subdivision with $O(\log_B n\log\log_B n)$ query cost and $O(\log_B n\log\log_B n)$ amortized update cost.  The computational  model used in this paper is the standard external memory model~\cite{AggarwalV88}. 

\section{Overview}
\label{sec:overview}
\subsection{Overall Structure}
\label{sec:overall}

As in the previous works, we concentrate on answering \emph{vertical ray shooting} queries. The successor segment of a point $q$ in a set $S$ of non-intersecting segments is the first segment that is hit by a ray emanating from $q$ in the $+y$-direction.  Symmetrically, the predecessor segment of $q$ in $S$ is the first segment hit by a ray emanating from $q$ in the $-y$ direction. 
A vertical ray shooting query for a point $q$ on a set of segments  $S$ asks for the successor segment of $q$ in $S$. If we know the successor segment or the predecessor segment of $q$ among all segments of a subdivision $\Pi$, then we can answer a point location query on $\Pi$ (i.e., identify the face of $\Pi$ containing $q$) in $O(\log_B n)$ I/Os~\cite{ArgeV04}. In the rest of this paper we will show how to answer vertical ray shooting queries  on a dynamic set of non-intersecting segments. 

Our base data structure is a variant of the segment tree. Let $\cS$ be a set of segments. We store a tree $\cT$ on $x$-coordinates of segment endpoints. Every leaf contains $\Theta(B)$ segment endpoints and every internal node has $r=\Theta(B^{\delta})$ children for $\delta=1/8$. Thus the height of $\cT$ is $O(\log_B n)$. We associate a vertical slab with every node $u$ of $\cT$.  The slab of the root node is $[x_{\min},x_{\max}]\times \mathbb{R}$, where $x_{\min}$ and $x_{\max}$ denote the $x$-coordinates of the leftmost and the rightmost segment endpoints. The slab of an internal node $u$ is divided into $\Theta(B^{\delta})$ slabs that correspond to the children of $u$.  A segment $s$ \emph{spans} the slab of a node $u$ (or simply spans $u$) if it crosses its vertical boundaries.

A segment $s$ is assigned to  an internal  node $u$, if $s$ spans at least one child $u_i$ of $u$ but does not span  $u$. We assign $s$ to a leaf node $\ell$ if at least one endpoint of $s$ is stored in $\ell$. 
All segments assigned to a node $u$ are trimmed to slab boundaries of children and stored in a \emph{multi-slab} data structure $C(u)$: Suppose that  a segment $s$ is assigned to $u$ and it spans the children $u_f$, $\ldots$, $u_l$ of $u$. Then we store the segment $s_u=[p_f,p_l]$ in $C(u)$, where $p_f$ is the point where $s$ intersect the left slab boundary of $u_f$ and $p_l$ is the point where $s$ intersects the right boundary of $u_l$. See Fig.~\ref{fig:slabs}. Each segment is assigned to $O(\log_B n)$ nodes of $\cT$.

In order to answer a vertical ray shooting query for a point $q$, we identify the leaf $\ell$ such that the slab of $\ell$ contains $q$. Then we visit all nodes $u$ on the path $\pi_{\ell}$ from the root of $\cT$ to $\ell$ and answer vertical ray shooting queries in  multi-slab structures $C(u)$. 
\begin{figure}[tb]
  \centering
  \includegraphics[width=.45\textwidth]{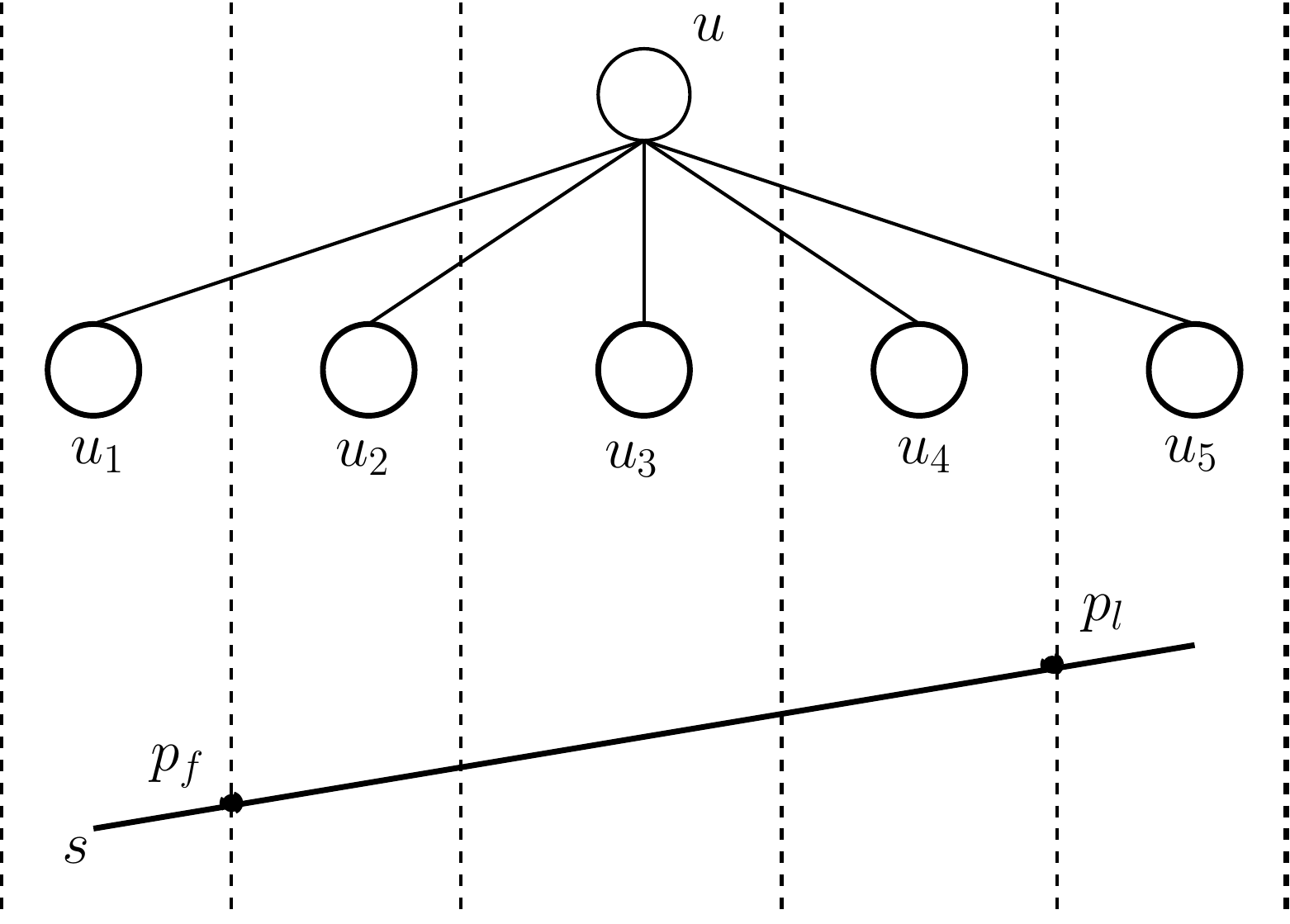}
  \caption{Segment $s$ is assigned to node $u$. The trimmed segment $[p_f,p_l]$ is stored in $C(u)$.}
  \label{fig:slabs}
\end{figure}

\subsection{Our Approach}
\label{sec:challenges}
Thus our goal is to answer $O(\log_B n)$ ray shooting queries in multi-slab structures along a path in the segment tree $\cT$  with as few  I/Os as possible. Segments stored in a multi-slab are not comparable in the general case; see Fig.~\ref{fig:order}. It is possible to impose a total order $\prec$ on all segments in the following sense: let $l$ be a vertical line  that intersects segments $s_1$ and $s_2$; if the intersection of $l$ with $s_1$ is above the intersection of $l$ with $s_2$, then  $s_2\prec s_1$. We can find such a total order in $O((K/B)\log_{M/B}K)$ I/Os, where $K$ is the number of segments~\cite[Lemma 3]{ArgeVV07}. But this ordering is not stable under updates: even a single deletion and a single insertion can lead to significant changes in the order of segments. See Fig.~\ref{fig:order}. 
Therefore it is hard to apply standard techniques, such as fractional cascading~\cite{ChazelleG,MehlhornN}, in order to speed-up ray shooting queries. Previous external-memory solutions in~\cite{AgarwalABV99,ArgeBR12} essentially perform $O(\log_B n)$ independent searches in the nodes of a segment tree or an interval tree in order to answer a query. Each search takes $O(\log_B n)$ I/Os, hence the total query cost is $O(\log_B^2 n)$.  

Internal memory data structures achieve $O(\log n)$  query cost using dynamic fractional cascading~\cite{ChazelleG,MehlhornN}.  
Essentially the difference with external memory is as follows: since we aim for $O(\log_2 n)$ query cost in internal memory, we can afford to use base tree $\cT$ with small node degree. In this special case the segments stored in sets $C(u)$, $u\in \cT$, can be ordered resp.\ divided into a small number of ordered sets. When the order of segments in $C(u)$ is known, we can apply the fractional cascading technique~\cite{ChazelleG,MehlhornN} to speed up queries. 
Unfortunately dynamic fractional cascading  does not work in the case when the total order of segments in  $C(u)$ is not known. 
Hence we cannot use previous internal memory solutions of the point location problem ~\cite{ChengJ92,BaumgartenJM94,ArgeBG06,ChanN15} to decrease the query cost in  external memory.

\begin{figure}[tb]
  \centering
  \begin{minipage}{.5\textwidth}
    \includegraphics[width=.6\linewidth,page=1]{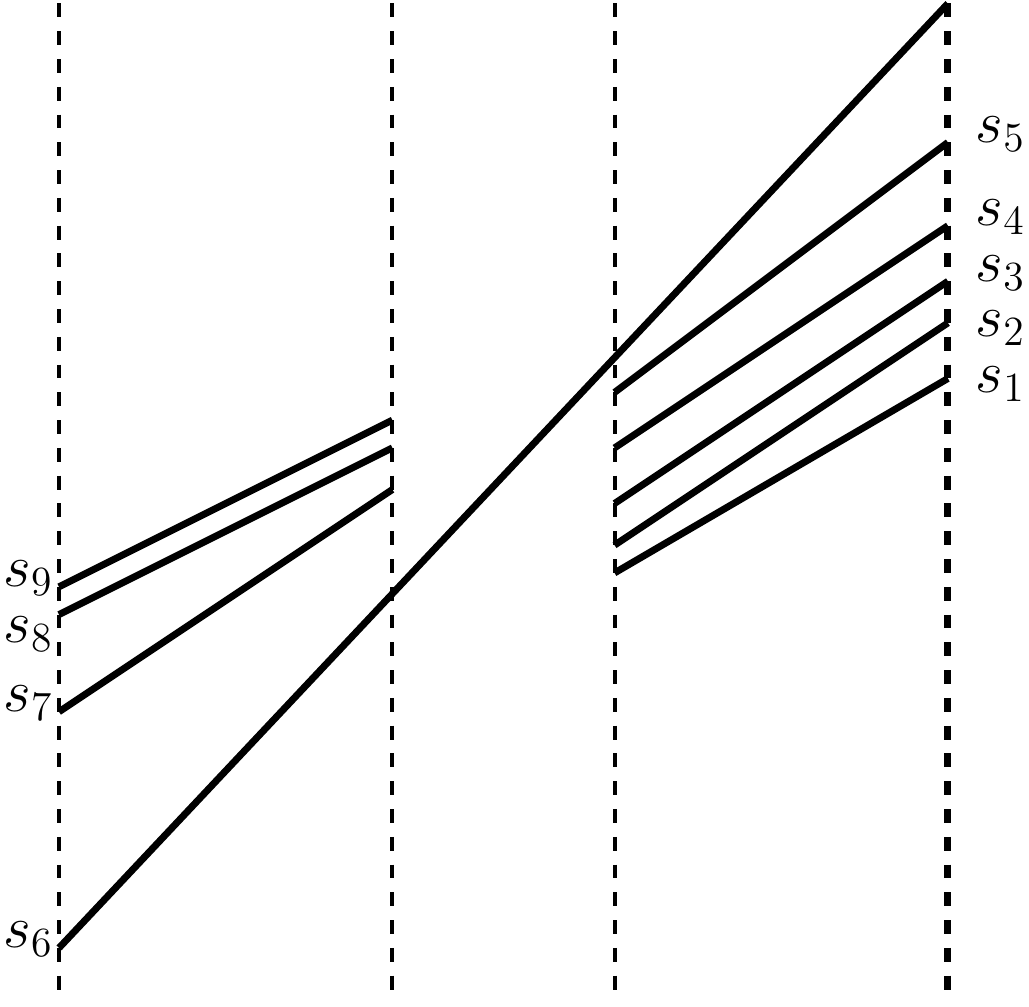}
  \end{minipage}%
  \begin{minipage}{.5\textwidth}
    \centering
    \includegraphics[width=.6\linewidth,page=4]{segment-order}
  \end{minipage}
  \caption{Left: example of segment order in a multi-slab; $s_1\!\!\prec\! s_2\!\!\prec\! s_3\!\!\prec\! s_4\!\!\prec\! s_5\!\!\prec\! s_6\!\!\prec\! s_7\!\!\prec\! s_8\!\!\prec\! s_9$. Right: a deletion and an insertion of one new segment in a multi-slab changes the order of segments to $s_7\!\!\prec\! s_8\!\!\prec\! s_9\!\!\prec\! s_0\!\!\prec\! s_1\!\!\prec\! s_2\!\!\prec\! s_3\!\!\prec\! s_4\!\!\prec\! s_5$.}
  \label{fig:order}
\end{figure}
In this paper we propose a different approach. Searching  in a multi-slab structure $C(u)$ is based on a weighted search among segments of $C(u)$. Weights of segments are chosen in such way that the total cost of searching in all multi-slab structures along a path $\pi$ is logarithmic. We also use fractional cascading, but this technique plays an auxiliary role: we apply fractional cascading to compute  the weights of segments and to navigate between the tree nodes. Interestingly, fractional cascading is usually combined with the union-split-find data structure, which is not used in our construction. 

 This paper is structured as follows.
In Section~\ref{sec:plstat} we show how our new technique, that will be henceforth called weighted telescoping search, can be used to solve the static vertical ray shooting problem. Next we turn to the dynamic case. In our exposition we assume, for simplicity, that the set of segment $x$-coordinates is fixed, i.e., the tree $\cT$  does not change. We also assume that the block size $B$ is sufficiently large, $B> \log^8 n$.  We show how our static data structure from Section~\ref{sec:plstat} can be modified to support insertions in Section~\ref{sec:pldyn}. To maintain the order of segments in a multi-slab under insertions we pursue the following strategy: when a new segment is inserted into the multi-slab structure $C(u)$, we split it into a number of \emph{unit} segments, such that every unit segment spans exactly one child of $u$. Unit segments can be inserted into a multi-slab so that the order of other segments is not affected. The number of unit segments per inserted segment can be large; however we can   use buffering  to reduce the cost of updates.\footnote{As a side remark, this approach works with weighted telescoping search, but it would not work with the standard fractional cascading used in internal-memory solutions~\cite{ChengJ92,BaumgartenJM94,ArgeBG06,ChanN15}. The latter technique relies on a union-split-find data structure (USF)  and it is not known how to combine buffering with USF.}  We need to make some further changes in our data structure in order to support deletions; the fully-dynamic solution for large $B$ is described in Section~\ref{sec:plfullydyn}. The main result of Section~\ref{sec:plfullydyn}, summed up in Lemma~\ref{lemma:slowupdate}, is the data structure that answers queries in $O(\log_B n\log\log_B n)$ I/Os; insertions and deletions are supported in $O(\log^2_B n)$ and $O(\log_B n)$
amortized I/Os respectively. \longver{We show how to reduce the cost of insertions and the space usage in Sections~\ref{sec:insfast} and Appendix~\ref{sec:space} respectively. We address some missing technical details and consider the case of small block size $B$ in Section~\ref{sec:smallB}.  The special case of vertical ray shooting among horizontal segments is studied in Appendix~\ref{sec:horiz}. 
Appendix~\ref{sec:tele} provides an alternative introduction to the weighted telescoping search by explaining how this technique works in a simplified scenario.  This section  is not used in the rest of the paper; the sole purpose of  Appendix~\ref{sec:tele} is to provide an additional explanation for the weighted telescoping search.}

\section{Ray Shooting: Static Structure}
\label{sec:plstat}
In this section we show how the weighted telescoping search can be used to solve the static point location problem. 
Let $\cT$ be the  tree, defined in Section~\ref{sec:overall}, with node degree $r=B^{\delta}$ for $\delta=1/8$.  Let $C(u)$ be the set of segments that span at least one child of $u$ but do not span $u$. 

\myparagraph{Augmented Catalogs.}  We keep augmented catalogs $AC(u)\supset C(u)$ in every node $u$.  Each $AC(u)$ is divided into subsets $AC_{ij}(u)$ for $1\le i\le j\le r$; $AC_{ij}(u)$ contains segments that span children $u_i$, $\ldots$, $u_j$ of $u$ and only those children.  Augmented catalogs $AC(u)$ satisfy the following properties: 
\begin{itemize*} 
  \item[(i)] If a segment $s\in (AC(u)\setminus C(u))$, then $s\in C(v)$ for an ancestor $v$ of $u$ and $s$ spans $u$.
\item[(ii)] Let  $E_i(u)=AC(u)\cap AC(u_i)$ for a child $u_i$ of $u$.  For any $f$ and $l$, $f\le i \le l$, there are at most $d=O(r^4)$ elements of $AC_{fl}(u)$ between any two consecutive elements of $E_i(u)$.
\item[(iii)] If $i\not=j$, then $E_i(u)\cap E_j(u)=\emptyset$.
\end{itemize*}

Elements of $E_i(u)$ for some $1\le i\le r$ will be called down-bridges; elements of the set $UP(u)=AC(u)\cap AC(\parent(u))$, where $\parent(u)$ denotes the parent node of $u$, 
are called up-bridges. We will say that a sub-list of a catalog $AC(u)$ bounded by two up-bridges is a \emph{portion} of  $AC(u)$.  
We refer to e.g., \cite{ArgeBG06} or \cite{ChanN15} for an explanation how we can construct and maintain $AC(u)$. We assume  in this section that all segments in every catalog $AC(u)$ are ordered. We can easily order a set $AC_{fl}(u)$ or  any set of segments that cross the same vertical line $\ell$: the order of segments  is determined by ($y$-coordinates of) intersection points of segments and $\ell$. Therefore we will speak of e.g., the largest/smallest segments in such a set. 
 


\myparagraph{Element  weights.}  We assign the weight to each element of $AC(u)$ in a bottom-to-top manner: All segments in a set $AC(\ell)$ for every leaf node $\ell$ are assigned weight $1$. Consider a segment $s\in AC_{fl}(u)$, i.e., a segment that spans children $u_f$, $\ldots$, $u_l$ of some internal node $u$. For $f\le i\le l$ let $s_1$ denote the largest bridge in $E_i(u)$ that is (strictly) smaller than  $s$ and let $s_2$ denote the smallest bridge in $E_i(u)$ that is (strictly) larger than $s$; we  let $W(s_1,s_2,u_i)=\sum_{s_1<  s' < s_2}weight(s',u_i)$, where the sum is over all segments $s'\in AC(u_i)$  and $weight_i(s,u)=W(s_1,s_2,u_i)/d$. See Fig.~\ref{fig:segment-weights} for an example. We set $weight(s,u)=\sum_{i=f}^l weight_i(s,u)$. We keep a weighted search tree for every portion $\cP(u)$ of the list $AC(u)$ By a slight misuse of notation this tree will also be denoted by $\cP(u)$. Thus every catalog $AC(u)$ is stored in a forest of weighted trees $\cP_j(u)$ where every tree corresponds to a portion of $AC(u)$~\footnote{In most cases we will  omit the subindex and will speak of a weighted tree $\cP(u)$ because it will be clear from the context what portion of $AC(u)$ is used.}.  We also store  a data structure supporting finger searches on $AC(u)$. 
\begin{figure}[tb]
  \centering
  \begin{minipage}{.5\textwidth}
    \includegraphics[width=.7\linewidth,page=1]{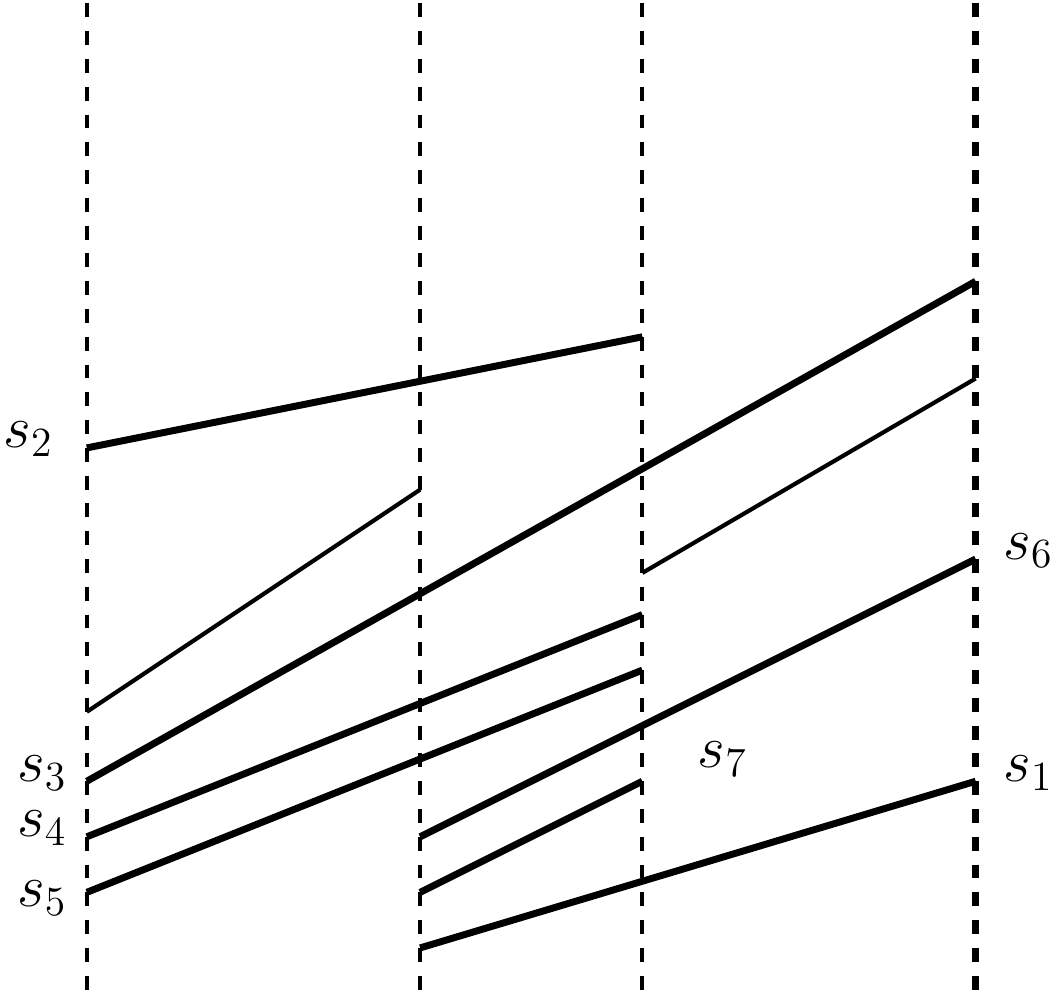}
  \end{minipage}%
  \begin{minipage}{.5\textwidth}
    \centering
    \includegraphics[width=.7\linewidth,page=2]{segment-weights}
  \end{minipage}
\caption{Computation of segment weights.  Left: segments $s_1$ and $s_2$ are down-bridges from $E_2(u)$.  Segments $s_3$, $s_4$, $s_5$, $s_6$, and $s_7$ are in $AC(u)\setminus E_2(u)$. 
For $3\le j \le 7$, $weight_2(s_j,u)=W(s_1,s_2,u_2)/d$. Right: portion of $AC(u_2)$ for the  child $u_2$ of $u$.  
$W(s_1,s_2,u_2)$ is equal to the total weight of all segments between $s_1$ and $s_2$. 
If $u_2$ is a leaf node, then $W(s_1,s_2,u_2)$ equals the total number of segments in $AC(u_2)$ that are situated between $s_2$ and $s_1$. }
\label{fig:segment-weights}
\end{figure}

\myparagraph{Weighted Trees.} Each weighted search tree is implemented as a biased $(a,b)$-tree  with parameters $a=B^{\delta}/2$ and $b=B^{\delta}$~\cite{BentST85,FeigenbaumT83}. 
The depth of a leaf $\lambda$ in a biased $(B^{\delta}/2,B^{\delta})$-tree is bounded by $O(\log_B(W/w_{\lambda}))$, where $w_{\lambda}$ is the weight of an element in the leaf $\lambda$ and $W$ is the total weight of all elements in the tree. Every internal node $\nu$ has 
$B^{\delta}$ children and every leaf holds $\Theta(B)$ segments\footnote{In the standard biased ($a,b$)-tree~\cite{BentST85,FeigenbaumT83}, every leaf holds one element. But we can modify it so that every leaf holds $\Theta(B)$ different elements (segments). The weight of a leaf $\lambda$ is the total weight of all segments stored in $\lambda$.}. 
In each internal node $\nu$ we keep $B^{3\delta}$ segments $\nu.\max_{jk}[i]$. For every child $\nu_i$ of $\nu$ and for all $j$ and $k$, $1\le j\le k\le r$, $\nu.\max_{jk}[i]$  is the highest segment from $AC_{jk}$ in the subtree of $\nu_i$ ; if there are no segments from $AC_{jk}$ in the subtree of $\nu_i$, then $\nu.\max_{jk}[i]=\mathrm{NULL}$. Using values of $\nu.\max$  we can find, for any node $\nu$ of the biased search tree, the child $\nu_i$ of $\nu$ that holds the successor segment of the query point $q$. Hence we can find the smallest segment $n(u)$ in a portion $\cP(u)$ that is above a query point $q$ in $O(\log_B (W_P/\omega_n))$ I/Os where $W_P$ is the total weight of all segments in $\cP(u)$ and $\omega_n$ is the weight of $n(u)$. 

\myparagraph{Additional Structures.}  When the segment $n(u)$ is known, we will need to find the bridges that are closest to $n(u)$ in order to continue the search.  
We keep a list $V_i(u)\subseteq AC(u)$ for each node $u$ and for every $i$, $1\le i\le r$.  $V_i(u)$ contains all segments of $E_i(u)$ and some additional segments chosen as follows:  $AC(u)$ is divided into groups so that each group consists of $\Theta(r^6)$ consecutive segments; the only exception is the last group in $AC(u)$ that contains $O(r^6)$ segments (here we use the fact that segments in $AC(u)$ are ordered).  We choose the constant in such way that every group but the last one contains $d\cdot r^2$ segments. If a group $G$ contains a segment that spans $u_i$, then we select the highest segment from $G$ that spans $u_i$ and the lowest segment from $G$ that spans $u_i$; we store both segments in $V_i$. See Fig.~\ref{fig:segments-Vi}. For every segment in $V_i$ we also store a pointer to its group in $AC(u)$. We keep $V_i$ in a B-tree that supports finger search queries. 
   
Suppose that we know the  successor segment $n(u)$ of a query point $q$ in $AC(u)$. We can find the successor  segment $b_n(u)$ of $q$ in $E_i(u)$  using $V_i$: Let $G$ denote the group that contains $n(u)$. We search in $G$ for the segment $b_n(u)\ge n(u)$ using finger search. If $b_n(u)$ is not in $G$, we consider the highest segment $s_1\in G$ that spans $u_i$. 
 By definition of $AC(u)$, there are at most $dr^2$ segments between $n(u)$ and $b_n(u)$. 
We can find $b_n(u)$ in $O(\log_B (dr^2))=O(1)$ I/Os by finger search on $V_i$ using $s_1$ as the finger. Using a similar procedure, we can find the highest bridge segment $b_p(u)\le n(u)$ in $E_i(u)$. 
\begin{figure}[tb]
  \centering
  \includegraphics[width=.3\linewidth]{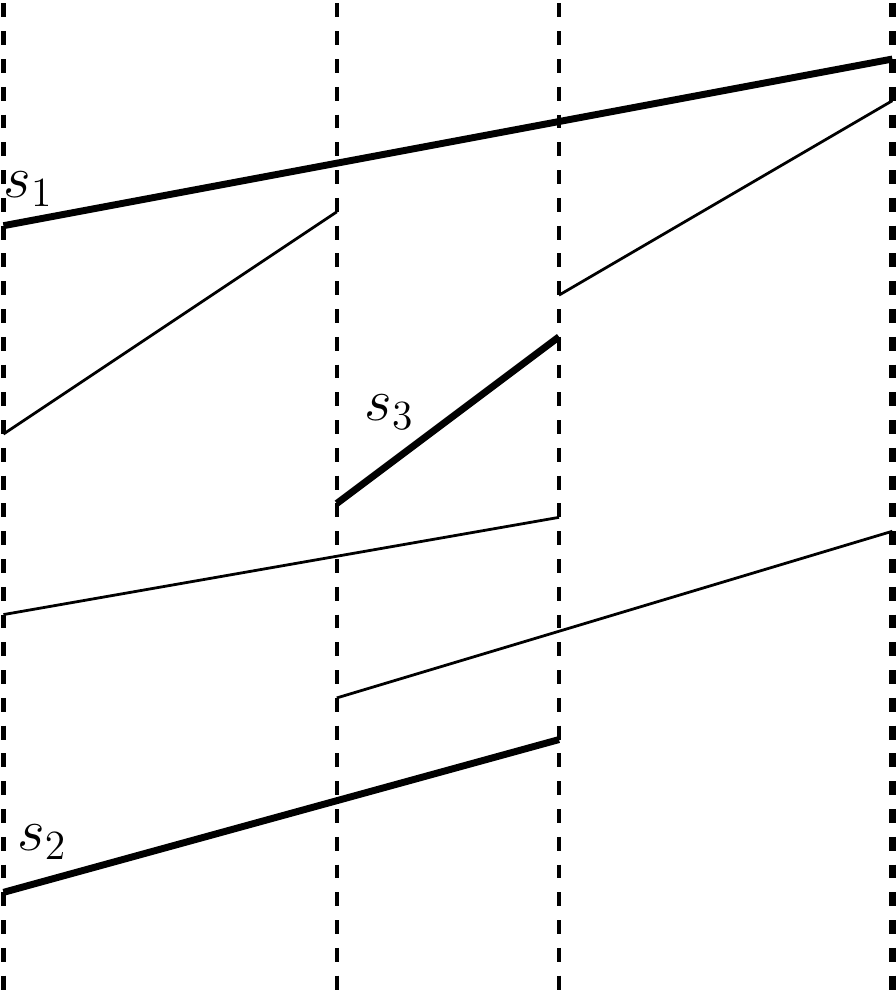}
\caption{Segments in a group. Segments $s_1$, $s_2$, and $s_3$ are stored in $V_2$: $s_1$ and $s_2$ are the highest and the lowest segments that span the second child $u_2$; $s_3$ is a bridge segment from $E_2(u)$.}
\label{fig:segments-Vi}
\end{figure}

\myparagraph{Queries.} A vertical ray shooting query for a point $q=(q_x,q_y)$ is answered as follows.  Let $\ell$ denote the leaf such that the slab of $\ell$ contains  $q$. We visit all nodes $v_0$, $v_1$, $\ldots$, $v_h$ on the root-to-leaf path $\pi(\ell)$ where $v_0$ is the root node and $v_h=\ell$. We find the segment $n(v_i)$  in every visited node, where  $n(v_i)$ is the successor segment of $q$ in $AC(v_i)$. Suppose that $v_{i+1}$ is the $j$-th child of $v_i$; $n(v_i)$ spans the $j$-th child of $v_i$.  First we search for  $n(v_0)$  in the  weighted tree of $AC(v_0)$. Next, using the list $V_j$, we identify the smallest bridge $b_n(v_0)\in E_j(v_0)$ such that $b_n(v_0)\ge n(v_0)$ and the largest bridge segment $b_p(v_0)\in E_j(v_0)$ such that $b_p(v_0)\le n(v_0)$. The index $j$ is chosen so that $v_1$ is the $j$-th child of $v_0$.  We execute the same operations in nodes $v_1$, $\ldots$, $v_h$. When we are in a node $v_i$ we consider the portion $\cP(v_i)$ between bridges $b_p(v_{i-1})$ and $b_n(v_{i-1})$; we search in the weighted tree of $\cP(v_i)$ for the successor segment $n(v_i)$ of $q$.  Then we identify the lowest  bridge $b_n(v_i)\ge n(v_i)$ and the highest bridge $b_p(v_i)\le n(v_i)$. When all $n(v_i)$ are computed, we find the lowest segment $n^*$ among $n(v_i)$. Since $\cup_{i=0}^h AC(v_i)=\cup_{i=0}^h C(v_i)$, $n^*$ is the successor segment of a query point $q$.

The cost of a ray shooting query can be estimated as follows.  Let $\omega_i$ denote the weight of $n(v_i)$.
Let $W_i$ denote the total weight of all segments of $\cP(v_i)$ (we assume that $\cP(v_0)=AC(v_0)$). 
Search for $n(v_i)$ in the weighted tree $\cP(v_i)$ takes  $O(\log_B(W_i/\omega_i))$ I/Os. By definition of weights, $\omega_i\ge W_{i+1}/d$. Hence 
\[\sum_{i=0}^h \log_B(W_i/\omega_i)= \log_B W_0+ \sum_{i=0}^{h-1} (\log_B W_{i+1} -\log_B\omega_i) - \log_B \omega_h\le \log_B(W_0/\omega_h) + 2(h+1)\log_B r.\] 
We have   $\omega_h=1$ and we will show below that   $W_0\le n$. Since $r=B^{\delta}$,    $h=O(\log_B n)$ and $\log_Br=O(1)$. Hence  the sum above can be bounded by $O(\log_B n)$. When $n(v_i)$ is  known, we can find $b_p(v_i)$ and $b_n(v_i)$ in $O(1)$ I/Os, as described above. Hence the total cost of answering a query is $O(\log_B n)$. 
Since every segment is stored in $O(\log_B n)$ lists $AC(u)$, the total space usage is $O(n\log_B n)$. 

It remains to prove that $W_0\le n$. We will show by induction that the total weight of all elements on every level of $\cT$ is bounded by $n$: Every element in a leaf node has weight  $1$; hence their total weight does not exceed $n$. Suppose that, for some $k\ge 1$, the total weight of all elements on level $k-1$ does not exceed $n$. Consider an arbitrary node $v$ on level $k$, let $v_1$, $\ldots$, $v_r$ be the children of $v$, and let $m_i$ denote the total weight of elements in $AC(v_i)$. Every element in $AC(v_i)$ contributes $1/d$ fraction of its weight to at most $d$ different elements in $AC(v)$. Hence $\sum_{e\in AL(v)}weight_i(v)\le m_i$ and the total weight of all elements in $AC(v)$ does not exceed $\sum_{i=1}^r m_i$. Hence, for any level $k\ge 1$, the total weight of $AC(v)$ for all nodes $v$ on  level $k$ does not exceed $n$. Hence the total weight of $AC(u_0)$ for the root node $u_0$ is also bounded by $n$.

\begin{lemma}
  \label{lemma:static}
There exists an $O(n\log_B n)$-space static data structure that 
supports point location queries on $n$ non-intersecting segments in $O(\log_B n)$ I/Os.
\end{lemma}
The result of Lemma~\ref{lemma:static} is not new. However we will show below that the data structure described in this section can be dynamized.

\section{Semi-Dynamic Ray Shooting for  $B\ge \log^8n$: Main Idea}
\label{sec:pldyn}
Now we turn to the dynamic problem. In Sections~\ref{sec:pldyn} and \ref{sec:plfullydyn} we will assume\footnote{Probably a smaller power of $\log$ can be used, but we consider $B\ge \log^8n$ to simplify the analysis.} that $B\ge\log^8n$.
\myparagraph{Overview.}
The main challenge in dynamizing the static data structure from Section~\ref{sec:plstat} is the order of segments. Deletions and insertions of segments can lead to significant changes in the segment order, as explained in Section~\ref{sec:overview}. However segment insertions within a slab are easy to handle in one special case. We will say that a segment $s\in AC(u)$  is a \emph{unit} segment if $s\in AC_{ii}(u)$ for some $1\le i \le r$. In other words a unit segment spans exactly one child $u_i$ of $u$. Let $L_i(u)=\cup_{f\le i\le l}AC_{fl}(u)$ denote the conceptual list of all segments that span $u_i$.  When a unit segment $s\in AC_{ii}(u)$ is inserted, we find the segments $s_p$ and $s_n$ that precede and follow $s$ in $L_i(u)$; we insert $s$ at an arbitrary position in $AC(u)$ so that $s_p< s < s_n$. It is easy to see that the correct order of segments is maintained: the correct order is maintained for the segments that span $u_i$ and other segments are not affected.

\begin{figure}[tb]
  \centering
  \begin{minipage}{.5\textwidth}
    \includegraphics[width=.6\linewidth,page=1]{segment-order}
  \end{minipage}%
  \begin{minipage}{.5\textwidth}
    \centering
    \includegraphics[width=.6\linewidth,page=5]{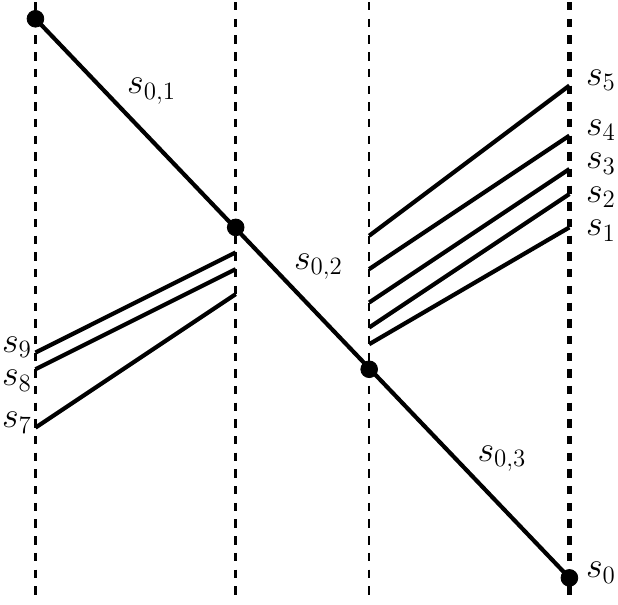}
  \end{minipage}
\caption{Example from Fig.~\ref{fig:order} revisited. Left: original segment order $s_1\!\!\prec\! s_2\!\!\prec\! s_3\!\!\prec\! s_4\!\!\prec\! s_5\!\!\prec\! s_6\!\!\prec\! s_7\!\!\prec\! s_8\!\!\prec\! s_9$.  Right: segment $s_6$ is deleted, the new inserted segment $s_0$ is split into unit segments. The new segment order is e.g., $s_{0,3}\!\!\prec\! s_1\!\!\prec\! s_2\!\!\prec\! s_4\!\!\prec\! s_5\!\!\prec\! s_{0,2}\!\!\prec\! s_7\!\!\prec\! s_8 \!\!\prec\! s_9\!\!\prec\! s_{0,1}$. Thus new unit segments are inserted, but the relative order of other segments does not change. 
 }
\label{fig:order-revisited}
\end{figure}

An arbitrary segment $s$ that is to be inserted into $AC(u)$ can be represented as $B^{\delta}$ unit segments. See Fig.~\ref{fig:order-revisited} for an example. However we cannot afford to spend $B^{\delta}$ operations for an insertion. To solve this problem, we use bufferization: when a segment is inserted, we split it into $B^{\delta}$ unit segments and insert them  into a buffer $\cB$. A complete description of the update procedure is given below.

\myparagraph{Buffered Insertions.} 
We distinguish between two categories of segments, \emph{old} segments and \emph{new} segments. We know the total order in the set of old segments in the portion $\cP(u)$ (and in the list $AC(u)$). New segments are represented as a union of up to $r$ unit segments. When the number of new segments in a portion $\cP(u)$ exceeds the  threshold that will be specified below, we re-build $\cP(u)$: we compute the order of old and new segments and declare all segments in $\cP(u)$ to be old.

As explained in Section~\ref{sec:plstat} every portion $\cP(u)$ of $AC(u)$ is stored in a biased search tree data structure. Each node of $\cP(u)$ has a buffer $\cB(\nu)$ that can store up to $B^{3\delta}$ segments. When a new segment is inserted into $\cP(u)$, we split it into unit segments and add them to the insertion buffer of $\nu_r$, where $\nu_r$ is the root node of $\cP(u)$. When the  buffer of an internal node $\nu$  is full, we \emph{flush} it, i.e., we move all segments from $\cB(\nu)$ to buffers in the children of $\nu$. We keep values $\nu.\max_{kj}[i]$, defined in Section~\ref{sec:plstat}, for all internal nodes $\nu$.  All $\nu.\max_{kl}[\cdot]$ and all segments in $\cB(\nu)$ fit into one block of memory; hence we can flush the buffer of an internal node in $O(B^{\delta})$ I/Os.  When the buffer of an internal node is flushed, we do not change the shape of the tree. When the buffer $\cB(\lambda)$ of a leaf node $\lambda$ is full, we insert segments from $\cB(\lambda)$ into the set of segments stored in $\lambda$. If necessary we create a new leaf $\lambda'$ and update the weights of $\lambda$ and $\lambda'$. We can update the biased search tree $\cP(u)$ in $O(\log n)$ time. We also update  data structures $V_i$ for $i=1$, $\ldots$, $r$. Since a leaf node contains the segments from at most two different groups, we can update all $V_i$ in $O(r)$ I/Os. The biased tree is updated in $O(\log n)$ I/Os. The total amortized cost of a segment insertion into a portion $\cP(u)$ is $O(1+\frac{\log n + r}{B^{3\delta}} + \frac{\log_B n}{B^{2\delta}})=O(1)$ because $B^{\delta}> \log n$.

When the number of new segments in $\cP(u)$ is equal to $n_{\mathtt{old}}/r$, where $n_{\mathtt{old}}$ is the number of old segments in $\cP(u)$, we rebuild $\cP(u)$. Using the method from~\cite{ArgeVV07}, we order \emph{all} segments in $\cP(u)$ and update the biased tree. Sorting of segments takes $O((n_{\mathtt{old}}/B)\log_{M/B} n_{\mathtt{old}})=o(n_{\mathtt{old}})$ I/Os.
We can re-build the weighted tree $\cP(u)$ in $O((n_{\mathtt{old}}/B^{3\delta})\log n_{\mathtt{old}})=o(n_{\mathtt{old}})$ I/Os 
by computing the weights of leaves and inserting the leaves into the new tree one-by-one.

When a new segment $s$ is inserted, we identify all nodes $u_i$ where $s$ must be stored. For every corresponding list $AC(u_i)$, we find the portion $\cP(u_i)$ where $s$ must be stored. This takes $O(\log_B^2 n)$ I/Os in total. Then we insert the trimmed segment $s$ into each portion as described above. The total insertion cost is $O(\log_B^2 n)$. Queries are supported in the same way as in the static data structure described in Section~\ref{sec:plstat}. The only difference is that biased tree nodes have associated buffers.  
Many technical  aspects are not addressed in this section. We fill in the missing details and provide the description of the data structure that also supports deletions in Section~\ref{sec:plfullydyn}.

\section{Ray Shooting for  $B\ge \log^8n$: Fully-Dynamic Structure}
\label{sec:plfullydyn}
Now we give a complete description of the fully-dynamic data structure for vertical ray shooting queries. Deletions are  also implemented using bufferization: deleted segments are inserted into  deletion buffers $\cD(\nu)$ that are kept in the  nodes of trees $\cP(u)$. Deletion buffers  are processed similarly to the insertion buffers.  There are, however, a number of details that were not addressed in the previous section. When a new bridge $E_i$ is inserted we need to change weights for a number of segments. 
When the segment $n(u)$ is found, we need to find the bridges $b_p(u)$ and $b_n(u)$. The complete solution that addresses all these issues is more involved. First, we apply weighted search only to segments from $E(u)=\cup_{i=1}^r E_i(u)$. We complete the search 
and find the successor segment in $AC(u)$ using some auxiliary sets stored in the nodes of $\cP(u)$. Second, we use a special data structure to find the bridges $b_p(u)$ and $b_n(u)$. We start by describing the changed structure of weighted trees $\cP(u)$. 

Segments stored in the leaves of $\cP(u)$ are divided into weighted and unweighted segments. Weighted segments are segments from $E(u)$, i.e., weighted segments are used as down-bridges. All other segments are unweighted. Every leaf contains $\Theta(r^2)$ weighted segments. There are at $\Omega(r^2)$ and $O(r^4)$ unweighted segments between any two weighted segments. Hence the total number of segments in a leaf is between $\Omega(r^4)$ and $O(r^6)$. Only weighted segments in a leaf have non-zero weights.  Weights of weighted segments are computed in the same way as explained in Section~\ref{sec:plstat}. Hence the  weight of a leaf $\lambda$ is the total weight of all weighted segments in $\lambda$.  The search for a successor of $q$ in $\cP(u)$ is organized in such way that it ends in the leaf holding the successor of $q$ in $E(u)$. Then we can find the successor of $q$ in $AC(u)$ using auxiliary data stored in the nodes of $\cP(u)$.

We keep the following auxiliary sets and buffers in nodes $\nu$ of every weighted tree $\cP(u)$. Let $AC_{fl}(u,\nu)$ denote the set of segments from $AC_{fl}(u)$ that are stored in leaf descendants of a node $\nu$.
\begin{itemize*}
\item[(i)] Sets $\Max_{fl}(\nu)$ and $\Min_{fl}(\nu)$ for all $f,l$ such that $1\le f\le l\le r$ and for all nodes $\nu$. $\Max_{fl}(\nu)$ ($\Min_{fl}(\nu)$) contains $\min(r^4,|AC_{fl}(u,\nu)|)$ highest (lowest) segments from $AC_{fl}(u,\nu)$.  
For every segment $s$ in sets $\Max_{fl}(\nu)$ and $\Min_{fl}(\nu)$ we record the index $i$ such that  $s\in E_i(u)$ (or NULL if $s$ is not a bridge segment). 
\item[(ii)] The set $\Nav(\nu)$ for an internal node $\nu$ is the union of all sets $\Max_{fl}(\nu_i)$ and $\Min_{fl}(\nu_i)$ for all children $\nu_i$ of $\nu$. 
\item[(iii)] The set $\Max'_{fl}(\nu)$, $1\le f\le l \le r$ contains  highest segments from $AC_{fl}(u,\nu)$ that are not stored in any set $\Max'(u,\mu)$ for an ancestor $\mu$ of $\nu$. Either $\Max'_{fl}(\nu)$ holds at least $r^4$ and at most $2r^4$ segments  or
$\Max'_{fl}(\nu)$ holds less than $r^4$ segments and $\Max'_{fl}(\rho)$ for all descendants $\rho$ of $\nu$ are empty. In other words, $\Max'_{fl}(\cdot)$ are organized as external priority search trees~\cite{ArgeSV99}.  The set $\Min'_{fl}(\nu)$ is defined in the same way with respect to the lowest segments. We use $\Max'$ and $\Min'$  to maintain sets $\Max$ and $\Min$.
\item[(iv)] Finally we keep an insertion buffer $\cB(\nu)$ and a deletion buffer $\cD(\nu)$ in every node $\nu$. 
\end{itemize*}

\myparagraph{Deletions.} If  an old segment $s$ is deleted, we insert it into the deletion buffer $\cD(\nu_R)$ of the root node $\nu_R$. If a new segment $s$ is deleted, we split $s$ into $O(r)$ unit segments and insert them into $\cD(\nu_R)$. When one or more segments are inserted into $\cD(\nu_r)$, we also update sets $\Max_{fl}(\nu_R)$ and $\Min_{fl}(\nu_R)$. For any node $\nu\in \cP(u)$, when the number of segments in $\cD(\nu)$ exceeds $r^3$, we flush both $\cD(\nu)$
and $\cB(\nu)$ using the following procedure. First we identify segments $s\in \cB(\nu)\cap \cD(\nu)$ and remove such $s$ from both $\cB(\nu)$ and $\cD(\nu)$. Next we move segments from $\cB(\nu)$ and $\cD(\nu)$ to buffers $\cB(\nu_i)$ and $\cD(\nu_i)$ in the children $\nu_i$ of $\nu$. For every child $\nu_i$ of $\nu$, first we update sets $Max'_{fl}(\nu_i)$ by removing segments from $\cD(\nu_i)$ (resp.\ inserting segments from $\cB(\nu_i)$) if necessary. Then we take care that the size of $\Max'_{fl}(\nu_i)$ is not too small. If some $\Max'_{fl}(\nu_i)$ contains less than $r^4$ segments and more than $0$ segments, we move up  segments from the children of $\nu_i$ into $\nu_i$, so that the total size of $\Max'_{fl}(\nu_i)$ becomes equal to $2r^4$ or all segments are moved from the corresponding sets $\Max'_{fl}(\cdot)$ in the children of $\nu_i$ into $\Max'_{fl}(\nu_i)$. We recursively update $\Max'_{fl}(\cdot)$ in each child of $\nu_i$ using the same procedure. 

Next, we update sets $\Max_{fl}(\nu_i)$. We compute $\cM_{fl}=\cup \Max'_{fl}(\mu)$ where the union is taken over all proper ancestors $\mu$ of $\nu$. Every segment in $\Max_{fl}(\nu)$ is either from $\Max'_{fl}(\nu)$ or from $\Max'_{fl}(\mu)$ for a proper ancestor $\mu$ of $\nu$. Hence we can compute all $\Max_{fl}(\nu_i)$ when $\cM_{fl}$ and $\Max'_{fl}(\nu_i)$ are known. Sets $\Min'_{fl}(\nu_i)$ and $\Min_{fl}(\nu_i)$ are updated in the same way. Finally we update the set $\Nav(\nu)$ by collecting segments from $\Max_{fl}(\nu_i)$ and $\Min_{fl}(\nu_i)$. 

All segments needed to re-compute sets after flushing buffers $\cD(\nu)$ and $\cB(\nu)$ fit into one block of space. Hence we can compute the set $\cM$ in $O(\log_B n)=O(r)$ I/Os and all sets in each node $\nu_i$ in $O(1)$ I/Os. The set $\Nav(\nu)$ is updated in $O(r)$ I/Os. Since each node has $O(r)$ children, the total  number of I/Os needed to flush a buffer is $O(r)$.  Every segment 
can be divided into up to $r$ unit segments and each unit segment can contribute to $\log_B n$ buffer flushes. Hence the total amortized cost per segment is $O(\frac{r^2\log_B n}{r^3})=O(1)$. We did not yet take into account the cost of refilling the buffers $\Max'$; using the analysis similar to the analysis in~\cite[Section 4]{Brodal16}, we can estimate the cost of re-filling $\Max'$ as $O(\frac{\log_B n}{r^3})=o(1)$. 

We do not store buffers in the leaf nodes. Let $S(\lambda)$ be the set of segments kept in a leaf $\lambda$ and let $S_W(\lambda)$ be the set of weighted segments stored in $\lambda$. When we move segments from $\cB(\nu)$ or $\cD(\nu)$ to its leaf child $\lambda$, we update $S(\lambda)$ accordingly. This operation changes the weight  of $\lambda$. Hence we need to update the weighted tree $\cP(u)$ in $O(\log n)$ I/Os. Sets $\Max_{fl}(\cdot)$ and $\Min_{fl}(\cdot)$ are also updated. 

After an insertion of new segments into a leaf node, we may have to insert or remove some bridges in $E_i(u)$ for $1\le i \le r$. When we insert a new bridge $b$ into $E_i(u)$, we must split some portion $\cP(u_i)$ into two new portions, $\cP_1(u_i)$ and $\cP_2(u_i)$.  Additionally we must change the weights  of the bridge segments in $E_i(u)$ that precede and follow $b$. The cost of splitting $\cP(u_i)$ is $O(\log n)$. We also need $O(\log n)$ I/Os to change the weights of two neighbor bridges. Hence the total cost of inserting a new bridge is $O(\log n)$. We insert a bridge at most once per $O(r)$ insertions into $AC(u)$ because every new segment is divided into up to $r$ unit segments. We remove a bridge at most once after $O(r)$  deletions. See~\cite{ChanN15} for the description of the method to maintain bridges in catalogs $AC(u)$.   Thus the total amortized cost incurred by a bridge insertion or deletion is $O(\frac{\log n}{r})=O(1)$.  

\myparagraph{Insertions.}
Insertions are executed in a similar way. A new inserted segment is split into $O(r)$ unit segments that are inserted into the buffer $\cB(\nu_R)$ for the root node $\nu_R$. The buffers and auxiliary sets are updated and flushed in the same way as in the case of deletions. When the number of new segments in some portion $\cP(u)$ is equal to $n_{\mathtt{old}}/r$, where $n_{\mathtt{old}}$ is the number of old segments in $\cP(u)$, we rebuild $\cP(u)$. As explained in Section~\ref{sec:pldyn}, rebuilding of $\cP(u)$ incurs an amortized cost of $o(1)$.

\myparagraph{Queries.}
The search for the successor segment $n(u)$ in the weighted tree $\cP(u)$ consists of two stages. Suppose that the query point $q$ is in the slab of the $i$-th child $u_i$ of $u$. First we find the successor $b_n(u)$ of $q$ in $E_i(u)$ by searching in $\cP(u)$. We traverse the path from the root to the leaf $\lambda_n$ holding $b_n(u)$. In every node $\nu$ we select its leftmost child $\nu_j$, such that $\Max_{fl}(\nu_j)$ for some $f\le i \le l$ contains a segment $s$ that is above $q$ and $s$ is not deleted (i.e., $s\not\in \cD(\mu)$ for all ancestors $\mu$ of $\nu$). The size of each set $\Max_{fl}(\nu_k)$ is larger than the total size of all $\cD(\mu)$ in all ancestors $\mu$ of $\nu$. Hence every $\Max_{fl}(\nu_i)$ contains some elements that are not deleted unless 
the set $C_{fl}(u,\nu_i)$ is empty.  Therefore we  select the correct child $\nu_j$ in every node.  Since $\cP(u)$ is a biased search tree~\cite{BentST85,FeigenbaumT83}, the total cost of finding the leaf $\lambda_n$ is bounded by $O(\log(W_P/\omega_{\lambda}))=O(\log (W_P/\omega_n))$ where $\omega_{\lambda}$ is the total weight of all segments in $\lambda_n$ and $\omega_n \le \omega_{\lambda}$ is the weight of the bridge segment $b_n(u)$. 

During the second stage we need to find the successor segment $n(u)$ of $q$ in $AC(u)$. The distance between $n(u)$ and $b_n(u)$ in $AC(u)$ can be arbitrarily large. Nevertheless $n(u)$ is stored in one of the sets $\Nav(\mu)$ for some ancestor $\mu$ of $\lambda_n$. Suppose that $n(u)$ is an unweighted segment stored in a leaf $\lambda'$ of $\cP(u)$ and let $\mu$ denote the lowest common ancestor of $\lambda$ and $\lambda'$. Let $\mu_k$ be the child of $\mu$ that is an ancestor of $\lambda'$. There are at most $r^4$ segments in $AC_{fl}(u)$ between $n(u)$ and $b_n(u)$. Hence, $n(u)$ is stored in the set $\Max_{fl}(\mu_k)$. Hence, $n(u)$ is also stored in $\Nav(\mu)$. We visit all ancestors $\mu$ of $\lambda_n$ and compute $\cD=\cup_{\mu}\cD(\mu)$. Then we visit all ancestors one more time and find the successor of $q$ in $\Nav(\mu)\setminus \cD$. The asymptotic query cost remains the same because we only visit the nodes between $\lambda_n$ and the root and each node is visited a constant number of times. 

We need to consider one additional special case. It is possible that there are no bridge segments $s\in E_i(u)$ stored in the leaves of $\cP(u)$. In this case there are at most $r^2$ segments in $AC_{fl}(u)$ for every pair $f,\,l$, satisfying $f\le i\le l$, stored in the leaves of $\cP(u)$. For each portion $\cP(u)$, if there are at most $r^2$ segments in $AC_{fl}(u)\cap \cP(u)$, we keep the list of all such segments. All such lists fit into one block of memory. We also keep the list of indexes $i$, such that  $E_i(u) \cap \cP(u)$ is empty.Suppose that we need to find the successor of $q$ and $\cP(u)\cap E_i(u)$ is empty. Then we simply examine all  segments in $AC_{fl}(u)\cap \cP(u)$ for all $f\le i\le l$ and find the successor of $q$ in $O(1)$ I/Os. 

When $n(u)$ is known, we need to find $b_p(u)$ and  $b_n(u)$, if $b_n(u)$ was not computed at the previous step. It is not always possible to find these bridges using $\cP(u)$ because $b_p(u)$ and $b_n(u)$ can be outside of $\cP(u)$.  To this end, we use the data structure for colored union-split-find problem on a list (list-CUSF) that will be described in \shlongver{the full version of this paper}{Section~\ref{sec:col}}. We keep the list $V(u)$ containing  all down-bridges from $E_i(u)$, for $1\le i\le r$, and all up-bridges from $UP(u)$. Each segment in $e\in V(u)$ is associated to an interval; a segment $e\in V_i(u)$ is associated 
to an interval $[i,i]$ and a segment from $UP(u)$ is associated to a dummy interval $[-1,-1]$. For any segment $e\in V(u)$ we can find the preceding/following segment associated to an interval $[i,i]$ for any $i$, $1\le i\le r$, in $O(\log\log_B n)$ I/Os.  Updates of $V(u)$ are supported in $O(\log \log_B n)$ I/Os. Since we insert or remove bridge segments once per $r^2$ updates, the  amortized cost of maintaining the list-CUSF structure is $O(1)$. 

\myparagraph{Summing up.}
By the same argument as in Section~\ref{sec:plstat}, weighted searches in all nodes take $O(\log_B n)$ I/Os in total.
Additionally we spend $(\log\log_B n)$ I/Os in every node with a query to list-CUSF. Thus the total query cost is $O(\log_B n\log\log_B n)$. 
When a segment is deleted, we remove it from $O(\log_B n)$ lists $AC(u)$ and from secondary structures (weighted trees etc.) in these nodes. The deletions take $O(1)$ I/Os per node or $O(\log_B n)$ I/Os in total. When a segment is inserted, it must be inserted into $O(\log_B n)$ lists $AC(u)$. We first have to spend $O(\log_B n)$ I/Os to find the portion $\cP(u)$ of each $AC(u)$ where it must be stored. When $\cP(u)$ is known, an insertion takes $O(1)$ amortized I/Os as described above. The total cost of an insertion is $O(\log^2_B n)$ I/Os. Since every segment is stored in $O(\log_B n)$ lists, the total space is $O(n\log_B n)$. 

\begin{lemma}
  \label{lemma:slowupdate}
     If $B>\log^8 n$, then there exists an $O(n\log_B n)$ space data structure that supports vertical ray shooting queries on a dynamic set of $n$ non-intersecting segments in $O(\log_B n\log\log_B n)$ I/Os. Insertions and deletions of segments are supported in $O(\log^2_B n)$  and $O(\log_B n)$ amortized I/Os respectively. 
\end{lemma}

\section{Faster Insertions}
\label{sec:insfast}
When a new segment $s$ is inserted into our data structure, we need to find the position of $s$ in $O(\log_B n)$ lists $AC(u)$ 
(to be precise, we need to know the portion  $\cP(u)$ of $AC(u)$ that contains $s$). When positions of $s$ in $AC(u)$ are known, we can finish the insertion in $O(\log_B n)$ I/Os. In order to speed-up  insertions, we use the multi-colored segment tree of Chan and Nekrich~\cite{ChanN15}. Segments in lists $C(u)$ are assigned colors $\chi$, so that the total number of different colors is $O(\log H)$ where $H=O(\log_B n)$ is the height of the segment tree.  Let $C_{\chi}(u)$ denote the set of segments of color 
$\chi$ in $C(u)$. We apply the technique of Sections~\ref{sec:plstat}-~\ref{sec:plfullydyn} to each color separately. That is, we create augmented lists $AC_{\chi}(u)$ and construct weighted search trees $\cP_{\chi}(u)$ for each color separately. The query cost is increased by factor $O(\log H)$, the number of colors.  The deletion cost is also increased by $O(\log H)$ factor because we update the data structure for each color separately. When a new segment $s$ is inserted, we insert it into some lists $AC_{\chi_i}(u_i)$ where $u_i$ is the node such that $s$ spans $u_i$ but does not span its parent and $\chi_i$ is some color (the same segment can be assigned different colors $\chi_i$ in different nodes $u_i$). We can find the position of $s$ in all $AC_{\chi_i}(u_i)$ with  $O(\log_B n\log H + H\cdot t_{\mathtt{usf}})=O(\log_B n\log\log_B n)$ I/Os where $t_{\mathtt{usf}}=O(\log\log_B n)$ is the query cost in a union-split-find data structure in the external memory model.  See~\cite{ChanN15} for a detailed description. 

  \begin{lemma}
\label{lemma:fastupd}
     If $B>\log^8 n$, then there exists an $O(n \log_B n)$ space data structure that supports vertical ray shooting  queries on a dynamic set of non-intersecting segments in $O(\log_B n(\log\log_B n)^2)$ I/Os. Insertions and deletions of segments can be  supported in $O(\log_B n\log\log_B n)$ amortized I/Os. 
  \end{lemma}

\section{Missing Details}
\label{sec:smallB}
Using the method from~\cite{ChanN15} we can reduce the space usage of our data structure to linear at the cost of increasing the query and update complexity by $O(\log\log_B n)$ factor.The resulting data structure supports queries in $O(\log_B n(\log\log_B n)^2)$ I/Os and updates in $O(\log_B n(\log\log_B n)^3)$ amortized I/Os.  \shlongver{Details will be provided in the full version.}{See Section~\ref{sec:space} for a more detailed description.} 

In our exposition we assumed for simplicity that the tree $\cT$ does not change, i.e., the set of $x$-coordinates of segment endpoints is fixed and known in advance. To support insertions of new $x$-coordinate, we can replace the static tree $\cT$ with a weight-balanced tree with node degree $\Theta(r)=\Theta(B^{\delta})$. 
We also assumed that the block size $B$ is large, $B>\log^8 n$. 
If $B\le \log^8 n$,  the linear-space internal memory data structure~\cite{ChanN15} achieves $O(\log n (\log\log n)^2)=O(\log_B n (\log\log_B n)^3)$ query cost  and $O(\log n\log\log n)=O(\log_B n (\log\log_B n)^2)$ update cost because 
$\log n=O(\log_B n\log\log_B n)$ and $\log\log n =O(\log\log_B n)$ for $B\le \log^8 n$. Thus we obtain our main result. 
\begin{theorem}
  \label{theor:main}
 There exists an $O(n)$ space data structure that supports vertical ray shooting queries on a dynamic set of $n$ non-intersecting segments in $O(\log_B n(\log\log_B n)^3)$ I/Os. Insertions and deletions of segments are supported in $O(\log_B n(\log\log_B n)^2)$ amortized I/Os. 
\end{theorem}

\bibliographystyle{plain}
\bibliography{dpl,extmem}

\longver{
\appendix
\section{Saving Space}
\label{sec:space}
 We use another method from~\cite{ChanN15} to reduce the space usage of the data structure in Lemma~\ref{lemma:fastupd} to linear. 

\begin{lemma}[\cite{ChanN15}, Lemma 3.1]\label{lemma:spaceCN15}
Consider a decomposable search problem, where
(i)~there is an $S(n)$-space fully dynamic
data structure with $Q(n)$ query cost and $U(n)$ update cost,
and (ii)~there is an $S_D(n)$-space deletion-only
data structure with $Q_D(n)$ query cost, $U_D(n)$ update cost,
and $P_D(n)$ preprocessing cost.
Then there is an $O(S(n/z) + S_D(n))$-space fully dynamic
data structure with $O(Q(n/z) + Q_D(n)\log z)$ query cost and
$O(U(n/z) + U_D(n)+(P_D(n)/n)\log z)$ amortized update cost
for any given parameter $z$ (assuming that $P_D(n)/n$ is
nondecreasing).
\end{lemma}

\begin{lemma}[\cite{ChanN15}, Lemma 3.1]\label{lemma:space2CN15}
If there is a deletion-only
data structure for vertical ray shooting queries
for $n$ \emph{horizontal} segments
with $S\orth(n)$ space, $Q\orth(n)$ query cost, $U\orth(n)$ update cost,
and $P\orth(n)$ preprocessing cost,
then there is a deletion-only
data structure for vertical ray shooting queries
for $n$ arbitrary non-intersecting segments
with $S_D(n)=S\orth(n)+O(n)$ space,
$Q_D(n)=Q\orth(n)+O(\log n)$ query cost,
$U_D(n)=U\orth(n)+O(1)$ update cost,
and $P_D(n)=P\orth(n)+O(n\log_B n)$ preprocessing cost.
\end{lemma}

The two above lemmata are obtained by a straightforward extension of Lemmata 3.1 and 3.2 from~\cite{ChanN15} to the external memory model. We will describe in Section~\ref{sec:horlinspace} a data structure that supports vertical ray shooting queries on a set of horizontal segments in $O(\log_B n \log\log_B n)$ I/Os and updates within the same amortized bounds. If we plug this result into Lemma~\ref{lemma:space2CN15}, we obtain a deletion-only data structure for ray shooting queries in a set of $n$ arbitrary non-intersecting segments with $Q_D(n)=O(\log_B n \log\log_B n)$ query cost, $U_D(n)=O(\log_B n\log\log_B n)$ deletion cost, and $P_D(n)=O(n\log_B n\log\log_B n)$ preprocessing cost.
Recall that the structure of Lemma~\ref{lemma:fastupd} has query cost  $Q(n)=O(\log_B n (\log\log_B n)^2)$, update cost $U(n)=O(\log_B n \log\log_B n)$ I/Os amortized, and space usage $O(n\log_B n)$. We apply Lemma~\ref{lemma:spaceCN15} to the structure of Lemma~\ref{lemma:fastupd} and the deletion-only structure described above. We  obtain the following lemma.
\begin{lemma}
  \label{lemma:linspace}
\tolerance=750
     If $B>\log^8 n$, then there exists an $O(n)$ space data structure that supports vertical ray shooting queries on a dynamic set of $n$ non-intersecting segments in $O(\log_B n(\log\log_B n)^3)$ I/Os. Insertions and deletions of segments are supported in 
$O(\log_B n(\log\log_B n)^2)$ amortized I/Os. 
\end{lemma}

\section{Ray Shooting on Horizontal Segments}
\label{sec:horiz}
In this section we describe a data structure for vertical ray shooting queries in a dynamic set of \emph{horizontal} segments.  A data structure for this problem can be used to answer dynamic point location queries in an orthogonal subdivision. The special case of horizontal segments is much simpler than the ray shooting among arbitrary segments because it is easy to maintain the order among horizontal segments. Our solution is based on a colored variant of the predecessor search, described in Section~\ref{sec:col}. 
We describe how this data structure can be combined with the segment tree to answer ray shooting queries in Section~\ref{sec:pl}. 
We show that  the space usage of our data structure can be reduced from $O(n\log_ B n)$ to $O(n)$ in Section~\ref{sec:horlinspace}. 
\subsection{Colored Predecessor Search in External Memory}
\label{sec:col}
In the colored predecessor searching problem, every element $e=(v_e,I_e)$ has a value  $v_e$ and color interval $I_e=[c_1,c_2]$, $1\leq c_1 \leq c_2 \leq C$. 
We assume that color intervals are disjoint for elements with the 
same value, i.e., if $v_a=v_b$ for two 
elements $a$ and $b$, then $I_a\cap I_b =\emptyset$. 
 The answer to a colored predecessor  query  $(v_q,\cC_q)$ for $v_q\in [1,V]$ and 
$\cC_q\subset [1,C]$ is the largest (with respect to its value $v_e$) 
element $e$, such that $v_e\leq v_q$ and $\cC_q\cap I_e\not=\emptyset$.
We say than an element $e$ is colored with a color $c$ if $c\in I_e$. 
First, we show how colored search queries for a small set of elements can be
 answered with a constant number of I/Os. Then, we will describe 
a data structure for an arbitrarily large set of elements.
\begin{lemma}\label{lemma:colsmall}
Let $\delta>0$ and $0<f\leq 1-\delta$.
 The colored predecessor searching problem for a set $S_l$, such that colors of  
elements belong to $[1,C]$ for $C\leq B^{f}$, 
 can be solved in $O(\log_B |S_l|)$ I/Os  
using a $O(|S_l|)$ space data structure that supports 
updates in $O(\log_B |S_l|)$ I/Os.
\end{lemma}
\begin{proof}
We sort the elements in $S_l$ by their values (elements with the same value 
are sorted by the smallest colors that belong to their color intervals)
and store them in the leaves of a tree $T_l$.  
Each leaf contains $\Theta(B)$ elements of  $S_l$  and 
each internal node has $\rho=\Theta(B^{\delta})$ children. 
We say that an internal node $u$ contains an element $s$ if $s$ is stored 
in a leaf descendant of $u$. In every internal 
node $u$, we store a table $R_u$: for each $c\in [1,C]$ and each 
$1\leq i\leq \rho$, $R_u[c,i]=1$ iff the $i$-th child of $u$ contains
 at least one 
element with color $c$. For every $(c,i)$ such that $R_u[c,i]=1$, 
we also store the maximal and the minimal elements colored with $c$ that 
belong to the $i$-th child of $v$.
Every table $R_u$ fits into $O(1)$ blocks. There are $O(|S_l|/B)$ internal 
nodes in $T_l$; hence, all $R_u$ 
use $O(|S_l|)$ words of space.

The search for $(x_q,c_q)$ starts at the root of $T_l$. In each visited 
node $u$ of $T$, we identify the rightmost child $u_i$ such that 
 $R_u[c_q,i]=1$ for some $c_q\in \cC_q$ and the minimal element colored 
with $c_q$ in $u_i$ is not greater than $x_q$. If such $u_i$ does not exist, 
then $S_l$ contains 
no element colored with $c_q\in \cC_q$ that is smaller than or equal to $x_q$.
Otherwise, the search continues in $u_i$. 
The height of $T_l$ is 
$O(\log_B |S_l|)$ and the search takes $O(\log_B |S_l|)$ I/Os. 

When a new element is inserted into $S_l$, we insert it into a leaf $u$ of 
$T_l$. Then, we visit all ancestors $u'$ of $u$ and update the tables $R_{u'}$.
The tree $T_l$ can be rebalanced in a standard way. Deletions are processed 
with a symmetric procedure.
\end{proof}

\begin{lemma}\label{lemma:colveb}
Let  $0<f<1/2$.
A colored searching problem for a set of $K$ elements, such that values 
of elements belong to  the universe 
$[1,V]$ and colors belong to the interval $[1,C]$ for 
$C\leq B^{f}$, 
 can be solved in $O(\log\log_B V)$ I/Os  using a 
$O(\max(V,K)\log\log_B V)$ space data structure that supports 
updates in $O(\log\log_B V)$ amortized I/Os.
\end{lemma}
\begin{proof}
To simplify the description, we introduce a new set of interval 
colors $\cM$: for 
each interval $[c_1,c_2]$, $1\leq c_1\leq c_2\leq C$, there is an interval 
color $c_{12}\in \cM$. Thus each interval color corresponds 
to a color interval $I_e$. 
Obviously, for each original color $c_i$ there is the interval color
 that corresponds to the interval $[c_i,c_i]$.  
An element $e$ is colored with a color 
from a set $\cC_q$ if and only if $e$ is colored with an interval 
color from a set 
$\cM_q=\{\,c_{ij}\,|\,\exists c\in\cC_q,\,c_i\leq c\leq c_j\,\}$.  
For every set of colors $\cC_q$ the equivalent set of interval colors 
$\cM_q$ can be constructed in $O(1)$ I/Os. 
While each element $e$ is colored with colors from an interval $I_e$, 
each $e$ is colored with only one interval color.


The data structure can be defined recursively.  If $\max(K,V)\leq B^2$, we can 
use the data structure of Lemma~\ref{lemma:colsmall} and answer queries 
in $O(1)$ I/Os.  If  $\max(K,V)>B^2$, then we 
divide the interval $[1,V]$ into subintervals of size\footnote{For ease of
 description we assume that $h$ is an integer.} 
$B^{h}$ for  $h=(\log_B V)/2$. 
The array $A[r]$ contains two tables $\min[i,j]$ and $\max[i,j]$ for each 
subinterval. The table $A[r].\min[i,j]$ ($A[r].\max[i,j])$ contains the minimal (maximal) element in the $r$-th subinterval colored with an interval
 color $c_{ij}$. 
If there is no such element in the $r$-th subinterval, then  
$A[r].\min[i,j]=A[r].\max[i,j]=\idtt{NULL}$.  
Let $S[r][i,j]$ be the set of elements whose values belong to the 
subinterval $[(r-1)B^{h}+1,rB^{h}]$ and whose interval color is 
$[i,j]$. 
If $|S[r][i,j]|\leq 2$, the values of elements from $S[r][i,j]$ are already 
stored in $A[r].\min[i,j]$ and $A[r].\max[i,j]$.
If $|S[r][i,j]|\geq 3$ for at least one pair $i,j$, we construct 
a recursively defined data structure $D[r]$ for $S[r]$. 
All values of elements in $D[r]$ are specified relative to the 
left end of the $r$-th subinterval; thus  values of all elements in 
$D[r]$ belong to $[1,B^{h}]$. 
The data structure $D_{top}$ supports colored predecessor searching in the 
array $A$: if $A[r].\min[i,j]\not=\idtt{NULL}$, then $D_{top}$ contains 
an element $e$ with $v_e=r$ and $I_e=\{c_{ij}\}$, i.e. $e$ is colored with 
an interval color $c_{ij}$. Values of elements in $D_{top}$ also belong 
to the interval $[1,B^{h}]$. 
All tables $\min[]$ and $\max[]$ in the array $A$ have  
$O(B^{\log_B V/2}\cdot B^{2f})$ entries. Therefore $A$ uses $O(\max(K,V))$ 
words of space.
 Since values of elements in 
$D_{top}$ and $D_r$ belong to $[1,B^{\log_B V/2}]$, there are at 
most $O(\log\log_B V)$ recursion levels. Hence, 
the total space usage is $O((\max(K,V)/B)\log\log_B V)$.

A query $(v_q,\cC_q)$ is processed as follows. 
Let $\cM_q$ be the set of interval colors that is equivalent to $\cC_q$.
If $K\leq B^2$, the query is answered in $O(1)$ I/Os by 
Lemma~\ref{lemma:colsmall}. 
Otherwise we check whether there is  at least one pair $i,j$ such 
that $A[v'_q].\min[i,j]\leq v_q$ and $[i,j]\cap \cC_q\not=\emptyset$; 
here $v'_q$ denotes the index of the subinterval that contains $v_q$.
This condition can be tested in $O(1)$ I/Os because each entry of $A[]$ 
fits into one block of memory. 
If such a pair exists, the search continues in the data structure 
$D[v'_q]$. Otherwise, we use the data structure $D_{top}$ to find the largest 
$r'<v'_q$ such that  $r'$ is colored with a color 
$c_{ij}\in \cM_q$. The answer to the query is the maximal element 
among all $A[r'].\max[i,j]$ such that $c_{ij}\in \cM_q$.

Suppose that $\max(V,K)\geq B^2$ and an element $e=(v_e,[i_e,j_e])$ is 
inserted. We update the values 
of $A[r].\min[i,j]$ and $A[r].\max[i,j]$ for $r=v_e/B^{h}$ if necessary. 
If  $|S[r][i,j]|\geq 3$ after the insertion of $e$, we insert 
an element $(v,[i,j])$,
such that $v\not=A[r].\min[i,j]$ and $v\not=A[r].\max[i,j]$ into $D[r]$. 
If $|S[r][i,j]|=1$ after the insertion of $e$, we insert an element $(r,[i,j])$
into the data structure $D_{top}$. Deletions are symmetric to insertions.
\end{proof}

In the colored union-split-find (CUSF) problem
 we put an additional restriction on 
queries: only queries $(v_q,\cC_q)$, such that there is an element $e\in S$ 
with value $v_e=v_q$, are allowed. Moreover, we assume that a pointer 
to such  an element $e\in S$ is provided with a query. When an element $e$ 
is deleted or inserted, we assume that the position of $e$ is known.
\begin{lemma}\label{lemma:colusf}
The CUSF problem for a set of $K$ elements and  $C\leq B^{1-f}$, 
$0<f< 1$, colors can be solved in $O(\log\log_B K)$ I/Os using a $O(K)$ space data structure that supports 
updates in $O(\log\log_B K)$ amortized I/Os. 
\end{lemma}
\begin{proof}
We can transform the general searching data structure into a CUSF data
 structure using the same principle as in~\cite{GioraK}. 
The set of $K$ elements is divided into chunks of size $\Theta(g)$. 
If $B\geq \log^2K$, then we set $g=B^{1+2f}$. Otherwise, we set $g=\log_B^4 K$.
We assign to each chunk $m$ an ordered  label $\lab(m)\in [1,O(K)]$: 
if  the chunk $m_1$ follows $m_2$, then $\lab(m_1)>\lab(m_2)$. 
Labels can be maintained according to the algorithm of~\cite{ItaiKR81,W92}: when 
a new chunk is inserted or when a chunk is deleted, $O(\log^2 K)$ labels 
must be changed.
The set of interval colors $\cM$ is defined exactly as in the proof of 
Lemma~\ref{lemma:colveb}.
 The data structure $D_v$ contains an element
$(\lab(m),c_{ij})$ if and only if the chunk $m$ contains an element 
with color interval $I_e=[c_i,c_j]$. By Lemma~\ref{lemma:colveb}, 
$D_v$ uses $O(K)$ words of space. We also store a data structure 
$D_m$ for each chunk that supports colored predecessor searching queries in this chunk. 
We implement the data structure $D_m$ as described in 
Lemma~\ref{lemma:colsmall},
so that colored  searching queries are supported in 
$O(\log_B g)=O(\log_B \log_B K)$ I/Os. All data structures $D_m$ 
use space $O(K)$.

Consider a query $(v_q,\cC_q)$. Suppose that some element 
$e\in S$ with $v_e=v_q$ 
belongs to a chunk $m_e$.  First, we find the  largest
chunk $m$, such that $m$ contains at least one interval color $c_{ij}$
and $[c_i,c_j]\cap \cC_q\not=\emptyset$.  The  set of colors 
$\cM_q=\{\,c_{ij}\,|\,\exists c\in\cC_q,\,c_i\leq c\leq c_j\,\}$ 
can be constructed with $O(1)$ I/Os. 
Using $D_v$, we can answer the colored search query 
for  $(m_e,\cM_q)$ and find the chunk $m$ in $O(\log\log_B K)$ I/Os. 
The largest element $e$ with $v_e\leq v_q$ and $I_e\cap \cC_q\not=\emptyset$ 
belongs to the chunk $m$. The cost of finding $e$ using the data 
structure for the chunk $m$ is $O(\log_B \log_B K)$; hence, a query can be 
answered with $O(\log\log_B K)$ I/Os.

When a new element $e$ is inserted, we insert it into a chunk $m_e$ 
in $O(1)$ I/Os. If the data structure $D_v$ does not 
contain the element $(\lab(m_e),I_e)$, then we insert this element 
into $D_v$ in $O(\log\log_B K)$ I/Os. 
If the number of elements in a block equals $2g$ we 
distribute the elements of $m$ between two chunks $m_1$ and $m_2$. 
It takes $O(g\log_B g)$ I/Os to delete the data structure $D_m$ and 
insert the elements of $m$ into data structures $D_{m_1}$ and $D_{m_2}$. 
We assign new labels to chunks $m_1$ and $m_2$ and update the set of 
labels. This leads to changing the labels of $O(\log^2_2 K)$ chunks. 
The data structure $D_v$ contains $O(B^{2f})$ labels for each 
chunk. Hence, the total number of updates in $D_v$ incurred by 
updating the set of labels is $O(B^{2f}\log^2 K)$. 
If $B> \log^2_2 K$, then $B^{1+2f}>B^{2f}\log^2 K$. 
If $B\leq \log^2K$, then $\log_B^4 K=\Omega( B^{2f}\log^2_2 K)$.
Since labels are changed after
 $\Theta(g)=\Omega(B^{2f}\log^2_2K)$ insertions, 
the amortized cost of an insertion is $O(\log\log_B K)$. 
Deletions are performed in the same way. Thus the total 
cost of an update is $O(\log\log_B K)$.
\end{proof}
We remark that, in fact,  the values of elements are not necessary. Using the same method, we can store a list of elements, such that each element $e$ in the list is assigned an interval of colors $I_e$. Given a pointer to a list element $e$ and a color interval $\cC_q$, we ask for the first element $e'$  that follows $e$ in the list  and $\cC_q\cap I_{e'}\not=\emptyset$. We will call this problem list CUSF.
\begin{lemma}
  \label{lemma:listCUCF}
The list CUSF problem for a set of $K$ elements and  $C\leq B^{1-f}$, $0<f< 1$, colors can be solved in $O(\log\log_B K)$ I/Os using a $O(K)$ space data structure that supports 
updates in $O(\log\log_B K)$ amortized I/Os. 
\end{lemma}
This result can be proved in the same way as Lemma~\ref{lemma:colusf}; we use it in Section~\ref{sec:plfullydyn}.

\subsection{Ray Shooting on Horizontal Segments}
\label{sec:pl}
\myparagraph{Structure.}
All segments are stored in a tree $\cT$ with node degree $B^c$ for some constant $c< 1/2$. The leaves of $\cT$ contain $x$-coordinates of segment endpoints; every leaf contains $\Theta(B)$ elements.  The tree is organized as a variant of the segment tree, in the same way as in Section~\ref{sec:overall}. 
We start by introducing some additional notation. The range of a leaf node $v_l$ is the interval $[a_l,b_l]$, where $a_l$ and $b_l$ are the minimal and the maximal values stored in $v_l$. 
The range of an internal node $v$ is the interval $[a,b]$, so that $a$ and $b$ are the minimal and maximal values stored in leaf descendants of $v$. 

We say that a segment  $s=(x_1,x_2;y)$ covers an interval $[a,b]$ if 
$x_1\leq a$ and $ b \leq x_2$; a segment covers $a$ if it covers an interval $[a,a]$.  Thus a  segment spans a node $v$ if it covers the range of $v$. We implement $\cT$ in the same way as before: a 
segment $s=(x_1,x_2;y)$ is associated 
with a node $v$ if and only if $s$ spans at least 
one child $v_i$ of $v$, but $s$ does not span the node $v$. 
Thus each segment is associated with $O(\log_B n)$ nodes. 

Let $C(v)$ be the set of segments associated with a node $v$. For simplicity, 
we sometimes will  not distinguish between a segment and its $y$-coordinate. 
For any point $q=(q_x,q_y)$, let $\pi$ denote search path for $q_x$ 
in the tree $\cT$. Each segment $s=(x_1,x_2;y_s)$, 
$x_1 \leq q_x\leq x_2$, is stored in a list $C(v)$, $v\in \pi$. 
If $s\in C(v)$ covers $q_x$ and $v$ is an internal node, then 
$s$ spans the child $v_i$ of $v$, $v_i\in \pi$.
Hence, we can identify the predecessor segment $s_q$ of $q$ by 
finding the highest segment $s\in C(v)$, $v\in \pi$, such that 
$s$ spans some node $v'\in \pi$ and $y_s\leq q_y$.

Our method is based on Lemma~\ref{lemma:colusf} and the fractional 
cascading technique~\cite{MehlhornN} applied to sets $C(v)$. We construct augmented catalogs 
$AC(v)\supset C(v)$ for all nodes $v$ of 
$\cT$. For a leaf 
node $v_l$, $AC(v_l)=C(v_l)$. Every list  $AC(v)$ is divided into 
groups $G_1(v),G_2(v),\ldots$, so that each group 
contains between $\log_B n/2$ and $2\log_B n$ segments. 
We guarantee that $AC(v)$ for an internal node $v$ contains one 
segment from a group $G_j(v_i)$ for every child $v_i$ of $v$ and 
every group $G_j(v_i)$. Moreover, $AC(v)$ contains all segments from $C(v)$. 
If copies of the same  segment $s$ are stored in augmented lists for 
a node $v$ and for a child 
$v_i$ of $v$, then the two copies of $s$ in $AC(v)$ and $AC(v_i)$ 
are connected by pointers, called bridges. Thus there are 
$O(\log_B n)$ elements of $AC(v_i)$ between any two consecutive bridges 
from $AC(v)$ to $AC(v_i)$.

The data structure $D(v)$ contains the colored set of ($y$-coordinates of) 
all segments in $AC(v)$: if a segment $s=(x_1,x_2;y_s)\in C(v)$ spans 
children $v_i,v_{i+1},\ldots,v_j$ of $v$, then an element $e_s=(y_s,[i,j])$ 
with value $y_s$ and colors $I_e=[i,j]$ is stored in $D(v)$; if a segment 
$s$ does not span any child $v_i$ of $v$ (i.e., $s$ belongs to 
$AC(v)\setminus C(v)$), then $e_s$ is colored with a dummy color $c_d$. 
The data structure $E(v)$ also contains a colored set of segments in $AC(v)$,
but segments are colored according to a different rule. 
All segments from $C(v)$ are colored with a dummy color $c_d$; for 
any segment $s\in AC(v)\setminus C(v)$, the set of colors for $s$
contains all values $i$ such that $s$ belongs to $AC(v_i)$ for a child 
$v_i$ of $v$. Both $D(v)$ and $E(v)$ are implemented as described in 
Lemma~\ref{lemma:colusf}.

\myparagraph{Queries.}
The search procedure visits all nodes on the path $\pi$ starting at the 
root. In every internal node $v\in \pi$, we identify the predecessor $s(v)$ of 
$q.y$ in $AC(v)$. Then we identify the highest segment $s'(v)$ in $AC(v)$ such 
that $s'(v)$ is below $s(v)$ and $s'(v)$ spans the child $v_i$ of $v$, 
$v_i\in \pi$.  Finally, we examine all segments in the leaf $v_l\in \pi$ 
and find the predecessor segment $s'(v_l)$ of $q$ stored in $C(v_l)$.
The predecessor segment of $q$ in $S$ is the highest segment 
among all $s'(v)$ for $v\in \pi$.

We can identify $s(v)$ for the root of $\cT$ in $O(\log_B n)$ I/Os
 using a standard B-tree. Suppose that $s(v)$ for a node $v$ is known. 
We will show how to find $s'(v)$ and $s(v_i)$ for the child $v_i\in \pi$ 
of $v$.
The highest segment $s'(v)\in AC(v)$,  such that $s'(v)$ is below 
$s(v)=(x_1,x_2;y_s)\in AC(v)$ and $s'(v)$ spans a child $v_i$ of $v$ can be
 found by answering the query $(y_s,i)$ to a data structure $D(v)$.
The segment $s(v_i)$ can be identified as follows. 
The highest segment $s_1(v)$, such that $s_1(v)$ is below $s(v)$ and 
$s_1(v)$ belongs to $AC(v_i)$ can be found in $O(\log\log_B n)$ I/Os 
using $E(v)$. Suppose that the copy of $s_1(v)$ belongs to a group 
$G_j(v_i)$ in $AC(v_i)$.
Since $s_1(v)$ is the highest segment below $q$ 
in $AC(v)\cap AC(v_i)$, $s(v_i)$ either belongs to the group $G_j(v_i)$ 
or to the next group $G_{j+1}(v_i)$. If we store $y$-coordinates of 
all segments from each   group $G_l$
in a B-tree, then we can search in  $G_l$  in $O(\log_B\log_B n)$ I/Os
because each $G_l$ contains $O(\log_B n)$ segments.
Thus the segment $s(v_i)$ can be found in $O(\log_B\log_B n)$ I/Os if 
$s_1(v)$ is known. 
Since the search procedure spends $O(\log\log_B n)$ I/Os in every node of 
$\pi$, the total cost of the search is $O(\log_B n\log\log_B n)$.

\myparagraph{Updates.}
Every segment belongs to $O(\log_B n)$ lists $C(v)$. Every insertion into 
a list $C(v)$ can be handled as follows. All nodes $v$ such that 
$s$ belongs to $C(v)$ are situated on at most two root-to-leaf paths 
$\pi$. We can identify positions of the $y$-coordinate of $s$ in 
all $AC(v)$ that belong to some path $\pi$ in $O(\log_B n \log\log_B n)$
I/Os as described in the search procedure. 

When we know the position of $s$ in a list $AC(v)$, we insert $s$ into 
data structures $D(v)$ and $E(v)$ in $O(\log\log_B n)$ I/Os. We 
also insert $s$ into the B-tree for its group $G_j$ in $C(v)$. If the 
number of elements in $G_j$ exceeds $2\log_B n$, $G_j$ is split into two 
groups $G'_j(v)$ and $G''_j(v)$. The list $AC(w)$ for the parent $w$ of $v$ 
already contains one element from either $G'_j(v)$ or $G''_j(v)$. 
A representative $s_r$ of another group must be inserted into $AC(w)$. 
The position of $s_r$ in $AC(w)$ can be found in $O(\log\log_B n)$ I/Os;
after that, $s_r$ is inserted into $AC(w)$ in $O(\log\log_B n)$ I/Os 
as described above. It can be shown that an insertion into a catalog 
$AC(v)$ leads to $O(1/\log_B n)$ insertions into augmented lists 
of ancestors of $v$~\cite{MehlhornN}. Hence, the total cost of an insertion is 
$O(\log_B n\log\log_B n)$. Deletions can be handled in the same way.
Since every segment is stored in $O(\log_B n)$ nodes, the total 
space usage is $O(n\log_B n)$.
\begin{lemma}\label{lemma:pl}
There exists a $O(n\log_B n)$ space data structure that supports 
ray shooting  queries on horizontal  segments in $O(\log_B n\log\log_B n)$
I/Os and updates in $O(\log_B n\log\log_B n)$ amortized I/Os. 
\end{lemma}
We can reduce the space usage to linear using the same approach as 
in~\cite{BaumgartenJM94,GioraK}. For completeness, we provide the proof in Section~\ref{sec:horlinspace}.
\begin{theorem}\label{theor:pl}
There exists a $O(n)$ space data structure that supports  ray shooting  queries on horizontal  segments in $O(\log_B n\log\log_B n)$
I/Os and updates in $O(\log_B n\log\log_B n)$ amortized I/Os. 
\end{theorem}

\subsection{Reducing Space to Linear}
\label{sec:horlinspace}
We follow the same approach as in~\cite{BaumgartenJM94,GioraK}. For a segment 
$s=(x_f,x_e;y_s)$, let $v_s$ be the lowest common ancestor of the leaves in 
which $x_f$ and $x_e$ are stored. The node $v_s$ is the lowest node 
such that the range of $v$ contains $[x_f,x_e]$, but $[x_f,x_e]$ does not 
span $v$.  Suppose that $s$ spans the children $v_i,\ldots,v_j$ of $v$.
We represent $s$ as a union of three segments: $s_m=(a_i,b_j;y_s)$, 
 $s_l=(x_f,a_i;y_s)$, and $s_r=(b_j,x_e;y_s)$, where 
$rng(v_i)=[a_i,b_i]$ and $rng(v_j)=[a_j,b_j]$. 
Segments $s_m$, $s_l$ and $s_r$ are stored in 
sets $S_m(v)$, $S_l(v_{i-1})$, and $S_r(v_{j+1})$ respectively. Let 
$\Pi_m=\cup_{v\in\cT} S_m(v)$, $\Pi_l=\cup_{v\in\cT} S_l(v)$, and
$\Pi_r=\cup_{v\in\cT} S_r(v)$. A ray shooting query can be answered 
by answering a ray shooting query for $\Pi_m$, $\Pi_r$, and $\Pi_l$. 

We can identify the predecessor segment of $q$ in $\Pi_m$ by storing 
all segments $s\in\Pi_m$ in the data structure of section~\ref{sec:pl}. 
Since every segment is stored only once the total space usage is $O(n)$. Now we describe the data structure for $\Pi_r$. 
A query on $\Pi_l$ can be answered using a symmetrically defined 
 data structure.

Each set $S_r(v)$ is divided into blocks $\cW$, so that each block contains 
$\Theta(\log_B n)$ segments. We denote by $win(\cW)$ the segment in a block 
$\cW$ with the largest $x$-coordinate of the right endpoint. 
The segments $win(\cW)$ for all blocks $\cW$ and all sets $S_r(u)$ 
are stored in a data structure $D_r$ implemented as in section~\ref{sec:pl}. 
Since the total number of segments in $D_r$ is $O(n/\log_Bn)$, 
$D_r$ needs $O(n)$ words. 
We denote by $Y_r(u)$ the set of $y$-coordinates of all segments in $S_r(u)$. 
The data structure $D_y$ is defined on all sets $Y_r(u)$. 
Let $u_i$ denote the nodes that lie on some root-to-leaf path of $\cT$.
Using $D_y$, we can search in all sets $Y_r(u_i)$ in $O(\log\log_B n)$ I/Os 
per node. 
To implement $D_y$, we apply the construction of section~\ref{sec:pl}, i.e., 
the augmented sets and CUSF structures, to sets $Y_r(u)$. 
Finally, for each block $\cW$ we store a data structure that 
supports queries on segments of $\cW$ in $O(\log_B (|\cW|))$ I/Os.
This data structure uses the fact that the left endpoints of all segments 
in $\cW$ lie on the same vertical line and is very similar to an external 
memory priority search tree~\cite{ArgeSV99}. 

A point location query for $q=(q_x,q_y)$ on $\Pi_r$ can be answered as
 follows. We start by identifying the successor segment $s^+$ of $q$ in $D_r$ and the predecessor segment $s^-$ of $q$ in $D_r$. 
Let $y^+$ and $y^-$ denote the $y$ coordinates of $s^+$ and $s^-$.  
Then, we use the data structure $D_y$ and find 
$y_1(u)=\spred(y^-,Y(u))$  and $y_2(u)=\ssucc(y^+,Y(u))$ for every node 
$u\in\pi$, where $\pi$ is the  search path  for $q_x$ in $\cT$. 
The following fact is proved in~\cite{GioraK}.
\begin{fact}
Let $\cW_1(u)$ and $\cW_2(u)$ be the blocks in $S_r(u)$ that contain 
segments with $y$-coordinates $y_1(u)$ and $y_2(u)$ respectively.
Let $s^*$ be the predecessor segment of $q$ in $\Pi_r$. 
Then $s^*$ belongs to a block $\cW_1(u)$ or $\cW_2(u)$ for some $u\in\pi$.
\end{fact}
We can complete the search by querying all $\cW_1(u)$ and $\cW_2(u)$, 
$u\in\pi$, in $O(\log_B n\log_B\log_B n)$ I/Os and selecting the highest segment
among all answers.

When a new segment is inserted, we identify the node $v_s$ in $O(\log_B n)$ 
I/Os. Insertion of $s_m$ into $\Pi_m$ is handled as in section~\ref{sec:pl}. 
Insertion of $s_r$ starts with identifying the 
child $v_r$ of $v_s$ such that $rng(v_r)$ intersects with $[x_f,x_e]$ but 
$s$ does not span $v_r$. Let $\cW_s$ be the block of $S_r(v_r)$ into which 
$s$ must be inserted. If the $x$-coordinate of the right endpoint of $s$ 
is larger than the $x$-coordinate of $win(\cW_s)$, then we 
remove $win(\cW_s)$ from $D_r$ and insert $s$ into $D_r$. 
We also insert the $y$-coordinate of $s$ into $D_y$. Finally, if the 
number of elements in $\cW_s$ after an insertion equals $2\log_B n$, then 
$\cW_s$ is split into two blocks that contain $\log_B n$ segments each. 
We can show using standard methods that the amortized cost of splitting a 
block  is $O(\log_B\log_B n)$. An insertion into $\Pi_l$ is symmetric.
Hence, the total cost of an insertion is $O(\log_B n \log\log_B n)$. 
Deletions  can be implemented in the same way.

\section{Weighted Telescoping Search: Simplified Scenario}
\label{sec:tele}
In this section we provide a simple alternative description of  our main technique, the weighted telescoping search. This section is not necessary in order to understand the material in the rest of this paper; the only purpose of this  section is to provide a simple description of the weighted telescoping search. To introduce our technique, we  digress from the point location problem and consider the following more simple scenario.  Suppose that we are given a balanced tree $\cT$ of node degree $r$ and we keep a  list (or catalog) $L(u)$ in every node $u$ of  $\cT$.
We assume that elements of $L(u)$ are numbers. The successor of  $q$ in a set $S$ is the smallest element $e$ in $S$ that is larger than or equal to $q$, $\ssucc(q,S)=\min\{\, e\in S\,|\, e\ge q\,\}$.   Suppose that we want to traverse a path $\pi(\ell)$ from the root  of $\cT$ to a leaf $\ell$ and search for the  successor of some $q$ in the union of all lists $L(u)$ along the path.  
In other words, for any element $q$ and any leaf $\ell$ we want to quickly find $\ssucc(q,\cup_{u\in \pi(\ell)}L(u))$.  This problem can be solved using the standard fractional cascading technique~\cite{ChazelleG} within the same time, but in this paper we describe an alternative solution. 
We also believe that this general technique  can be used  in other scenarios when the standard fractional cascading is hard to apply.  Unlike the rest of this paper, in this section we describe the solution for  the internal memory model.

Our solution is based on assigning weights to elements of $L(u)$ and maintaining a forest of weighted trees on each list $L(u)$ for every node $u\in \cT$. Roughly speaking, we choose the weights in such a way that  the weight of an element $e\in L(u)$ 
gives us an estimate on the number of elements $e'$, such that $e'$ is stored in $L(v)$ for some descendant $v$ of $u$ and $pred(e,L(u))\le e' \le e$.  We keep augmented catalogs $AL(u)\supseteq L(u)$ in order to compute and maintain element weights. 

\myparagraph{Augmented Lists.} We maintain an augmented catalog $AL(u)$ in every node $u$. Augmented catalogs are supersets of $L(u)$, $AL(u)\supseteq L(u)$, that satisfy the following properties:
\renewcommand{\labelenumi}{\theenumi}
\renewcommand{\theenumi}{\roman{enumi}}
\begin{itemize*}
\item[(i)] 
  If $e\in (AL(u)\setminus L(u))$, then $e\in L(v)$ for an ancestor $v$ of $u$.
\item[(ii)]
  Let a subset $E_i(u)$ of  $AL(u)$ be defined as $E_i(u)=AL(u)\cap AL(u_i)$ for a child $u_i$ of $u$.  There are at most $d=O(r^2)$ elements of $AL(u)$ between any two consecutive elements of $E_i(u)$.
\end{itemize*}
Elements of $E_i(u)$ for some $1\le i\le r$ will be called down-bridges; elements of the set $UP(u)=AL(u)\cap AL(\parent(u))$, where $\parent(u)$ denotes the parent node of $u$, 
are called up-bridges. We will say that a sub-list of a catalog $AL(u)$ bounded by two up-bridges is a \emph{portion} of  $AL(u)$.  
We can create and maintain augmented lists $AL(u)$ using the fractional cascading technique~\cite{ChazelleG,MehlhornN}. The main idea is to copy selected elements from $L(u)$ and store the copies in lists $AL(v)$ for descendants $v$ of $u$; see e.g.,~\cite{ChanN15} for a detailed description. If the same element $e$ is stored in lists $AL(u)$ and $AL(\parent(u))$, where $\parent(u)$ is the parent node of $u$, then we assume that there are pointers between instances of $e$ in $AL(u)$ and $AL(\parent(u))$.


\myparagraph{Element  weights.}  We assign the weight to each element of $AL(u)$ in a bottom-to-top manner: for a leaf node $\ell$ every element $e\in AL(\ell)$ is assigned  weight $1$. Consider an internal node $u$ with children $u_1$, $\ldots$, $u_r$. We associate values $weight_i(e,u)$ with each $e\in AL(u)$ for every $i$, $1\le i \le r$. Let $e_1$ and $e_2$ denote two consecutive bridge elements in $E_i(u)$.  Let $W(e_1,e_2,u)=\sum_{e_1\le e'\le e_2} weight(e',u)$ denote the total weight of all elements $e'\in AL(u)$ such that $e_1< e'\le e_2$.  Every element $e\in AL(u)$ that satisfies $e_1< e\le e_2$ is assigned the same value  of \[weight_i(e,u)=W(e_1,e_2,u_i)/d.\]
The weight of $e\in AL(u)$ is defined\footnote{We observe that the same element $e$ can be assigned different weights $weight(e,u)$ in different nodes $u$.} as $weight(e,u)=\sum_{i=1}^{r}weight_i(e,u)$.

\myparagraph{Telescoping  Search.}
Consider a sub-list $\cP(u)$ of the augmented catalog $AL(u)$ bounded by two up-bridges. All elements of $\cP(u)$ are stored in a weighted search tree satisfying the following property: The depth of a leaf holding an element $e\in \cP(u)$ is bounded by $O(\log(W_P/weight(e)))$ where $W_P=\sum_{e'\in\cP(u)}weight(e')$ is the total weight of all elements in $\cP(u)$. We can use e.g. biased search trees~\cite{BentST85,FeigenbaumT83} for this purpose. 

Now suppose that we want to find, for some number $q$ and some leaf $\ell$ of $\cT$,  the successor of $q$ in $\cup_{u\in \pi(\ell)}L(u)$ where $\pi(\ell)$ is the path 
from the root to $\ell$. We start in the root node $u_0$ and  identify  $n(u_0)=\ssucc(q,AL(u_0))$.  Suppose that $u_1$  is the $j$-th child of $u_0$. We find the largest down-bridge $b_p(u_0)\le n(u_0)$ and the smallest down-bridge $b_n(u_0)\ge n(u_0)$ where $b_p\in E_j(u_0)$ and $b_n\in E_j(u_0)$. We use finger search~\cite{GuibasMPR77} in $AL(u_0)$ with $n(u_0)$ as a finger to find 
$b_n(u_0)$ and $b_p(u_0)$ in $O(\log r)$ time. Next we identify the portion $\cP(u_1)$ bounded by $b_p(u_0)$ and $b_n(u_0)$. We find  $n(u_1)=\ssucc(q,AL(u_1))$ by searching in the weighted tree for $\cP(u_1)$. We then find the largest $b_p(u_1)\le n(u_1)$ and the smallest $b_n(u_1)\ge n(u_1)$ where $b_p\in E_g(u_1)$, $b_n\in E_g(u_1)$ and $u_2$ is the $g$-th child of $u_1$. Again we use finger search  and find $b_n(u_1)$, $b_p(u_1)$ in $O(\log r)$ time.  We continue in the same way until the leaf node is reached.  See Fig.~\ref{fig:telescoping}.

When we know $\ssucc(q,AL(u_i))$ for every node $u_i\in \pi(\ell)$, $n^*=\ssucc(q,\cup_{u_i\in\pi(\ell)} AL(u_i))$ can be computed. Every element $e\in AL(u)$ is either  from the set $L(u)$ or from the set $L(w)$ for some ancestor $w$ of $u$.  Hence $\cup_{u\in\pi(\ell)} AL(u)=\cup_{u\in\pi(\ell)} L(u)$.  Hence  $n^*$ is the successor of $q$ in $\cup_{u\in \pi(\ell)} L(u)$.

The total time can be estimated as follows. Let $\omega_i$ denote the weight of $n(u_i)$. We can find the element $e_n(u_0)$ in time $\log(W_0/\omega_0)$, where $W_0$ is the total weight of all elements in $AL(u_0)$.  Down-bridges $b_p(u_0)$ and $b_n(u_0)$ can be found in $O(\log r)$ time by finger search in $AL(u_0)$. When we know 
$b_p(u_i)$ and $b_n(u_i)$, we can compute $n(u_{i+1})$ in time $O(\log(W_{i+1}/\omega_{i+1}))$
where $W_{i+1}$ is the total weight of all elements in $\cP(u_{i+1})$. When $n(u_{i+1})$ is known, we can compute  $b_p(u_{i+1})$ and $b_n(u_{i+1})$ in $O(\log r)$ time. The total time needed to compute all $n(u_i)$ is $O(\sum_{i=0}^h \log(W_i/\omega_i))$ where $h$ is the tree height.  Since $\omega_i\ge W_{i+1}/r^2$, we have 
\[\sum_{i=0}^h \log (W_i/\omega_i)= \log W_0+ \sum_{i=0}^{h-1} (\log W_{i+1} -\log \omega_i) - \log \omega_h\le \log(W_0/\omega_h) + 2(h+1)\log r.\] 
By definition, $\omega_h=1$. We will show below that $W_0\le n$. Since $h=\log_r n$, $2(h+1)\log r=O(\log n)$ and the sum above can be bounded by $O(\log n)$. All finger searches also take $O(\log n)$ time. 

It remains to prove that $W_0\le n$. We will show by induction that the total weight of all elements on every level of $\cT$ is bounded by $n$: Every element in a leaf node has weight  $1$; hence their total weight does not exceed $n$. Suppose that, for some $k\ge 1$, the total weight of all elements on level $k-1$ does not exceed $n$. Consider an arbitrary node $v$ on level $k$, let $v_1$, $\ldots$, $v_r$ be the children of $v$, and let $m_i$ denote the total weight of elements in $AL(v_i)$. Every element in $AL(v_i)$ contributes $1/d$ fraction of its weight to at most $d$ different elements in $AL(v)$. Hence $\sum_{e\in AL(v)}weight_i(v)\le m_i$ and the total weight of all elements in $AL(v)$ does not exceed $\sum_{i=1}^r m_i$. Hence, for any level $k\ge 1$, the total weight of $AL(v)$ for all nodes $v$ on  level $k$ does not exceed $n$. Hence the total weight of $AL(u_0)$ for the root node $u_0$ is also bounded by $n$.  
\begin{figure}[tb]
  \begin{subfigure}{.3\textwidth}
    \includegraphics[width=.7\linewidth]{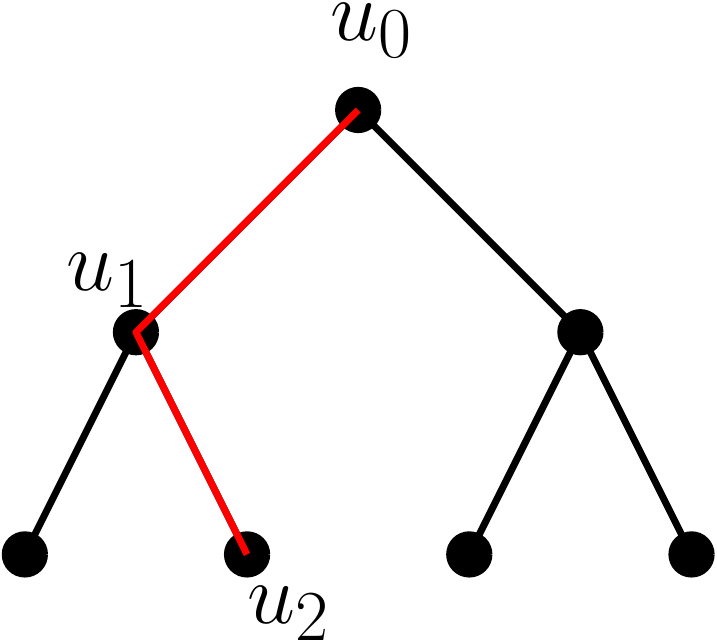}
  \end{subfigure}
  \begin{subfigure}{.7\textwidth}
     \centering
     \includegraphics[width=.7\linewidth]{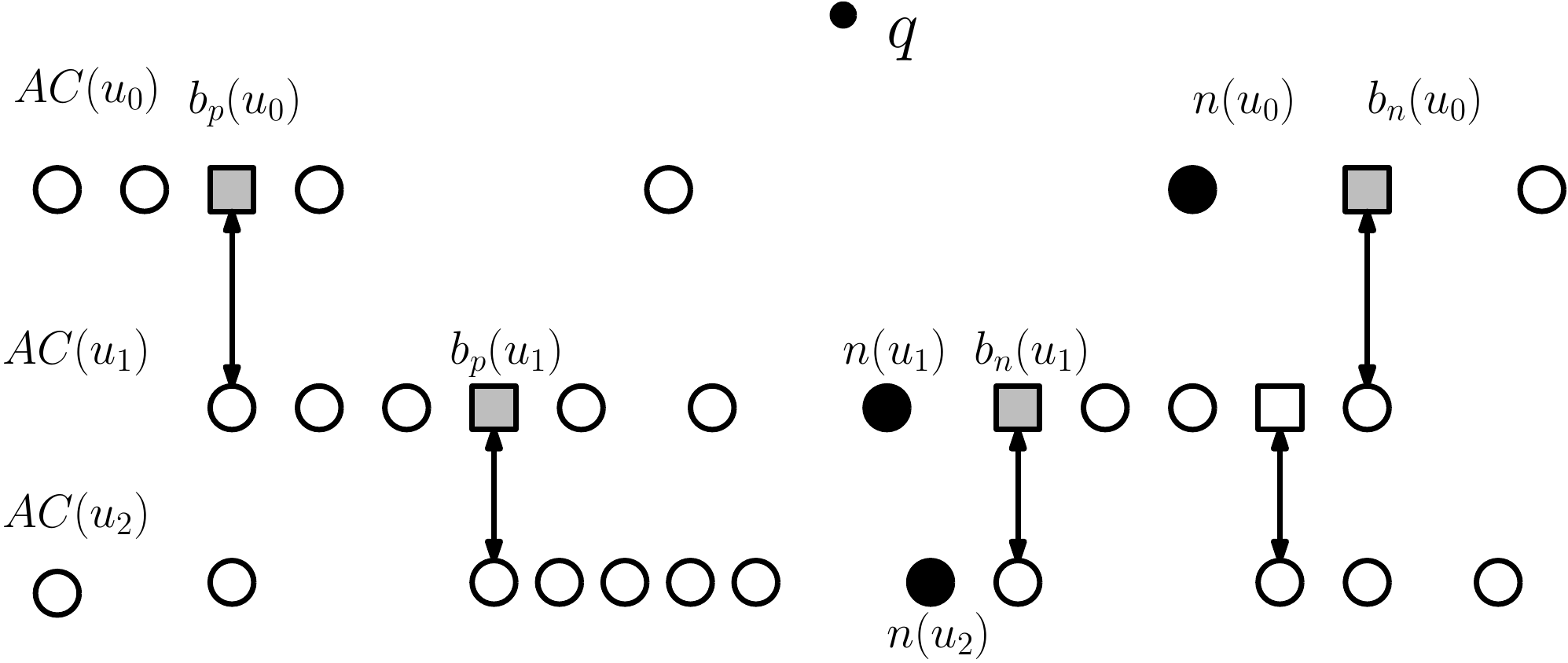}
  \end{subfigure}
  \caption{Telescoping Search. Left: tree $\cT$ with arity $r=2$ and a path $\pi=u_0,\,u_1,\,u_2$ in $\cT$. Right: Example of the telescoping search. We are searching for the predecessor of $q$ in $AL(u_0)\cup AL(u_1)\cup AL(u_2)$.   Down-bridges are shown with squares, other nodes are shown with circles.  $n(u_i)$ are shown with filled circles, $b_n(u_i)$ and $b_p(u_i)$ are shown with gray squares. Only relevant down-bridges and relevant parts of $AL(u_i)$ are shown.}
  \label{fig:telescoping}
\end{figure}
Thus we have shown the following result.
\begin{lemma}
  \label{lemma:tele1}
 Suppose that we store a sorted list $L(u)$ in every node $u$ of a balanced degree-$r$ tree $\cT$. Then it is possible to find $\ssucc(q,\cup_{u\in \pi}L(u))$ for any $q$ and for any root-to-leaf path $\pi$ in $O(\log n)$ time, where $n$ is the total size of all lists $L(u)$. The underlying data structure uses  space $O(n)$.
\end{lemma}
It is possible to extend the result of this section to the external memory model and to dynamize our data structure. However Lemma~\ref{lemma:tele1} cannot be used to answer vertical ray shooting queries because in the scenario of this Lemma lists $L(u)$ contain numbers.
}

\end{document}